\documentclass{lmcs} 
\pdfoutput=1

\usepackage[utf8]{inputenc}

\usepackage{lastpage}
\lmcsdoi{20}{1}{15}
\lmcsheading{}{\pageref{LastPage}}{}{}%
{Sep.~08,~2022}{Feb.~21,~2024}{}


\usepackage{hyperref}
\usepackage{amsmath,amssymb}
\usepackage{bm}
\usepackage{stmaryrd}
\usepackage{mathtools}
\usepackage[ruled, noend]{algorithm2e}
\usepackage{tikz}
\usetikzlibrary{petri,positioning,fit,decorations.pathmorphing,shapes}
\pgfdeclarelayer{back}
\pgfsetlayers{back,main}
\usepackage{macros}

\theoremstyle{plain}

\begin{document}

\title{Separators in Continuous Petri Nets}
\thanks{M.~Blondin was supported by a Discovery Grant from the Natural
  Sciences and Engineering Research Council of Canada (NSERC), and by
  the Fonds de recherche du Qu\'{e}bec -- Nature et technologies
  (FRQNT). J.~Esparza was supported by an ERC Advanced Grant (787367:
  PaVeS)}

\author[M.~Blondin]{Michael Blondin\lmcsorcid{0000-0003-2914-2734}}[a]
\author[J.~Esparza]{Javier Esparza\lmcsorcid{0000-0001-9862-4919}}[b]

\address{Universit\'{e} de Sherbroke, Sherbrooke, Canada} 
\email{michael.blondin@usherbrooke.ca}  

\address{Technical University of Munich, Munich, Germany}	
\email{esparza@in.tum.de}  




\begin{abstract}
  \noindent Leroux has proved that unreachability in Petri nets can be
  witnessed by a Presburger separator, i.e.\ if a marking $\msrc$
  cannot reach a marking $\mtgt$, then there is a formula $\varphi$ of
  Presburger arithmetic such that: $\varphi(\msrc)$ holds; $\varphi$
  is forward invariant, i.e., $\varphi(\vec{m})$ and $\vec{m}
  \rightarrow \vec{m}'$ imply $\varphi(\vec{m}'$); and $\neg
  \varphi(\mtgt)$ holds. While these separators could be used as
  explanations and as formal certificates of unreachability, this has
  not yet been the case due to their worst-case size, which is at
  least Ackermannian, and the complexity of checking that a formula is
  a separator, which is at least exponential (in the formula size).

  We show that, in continuous Petri nets, these two problems can be
  overcome. We introduce locally closed separators, and prove that:
  (a)~unreachability can be witnessed by a locally closed separator
  computable in polynomial time; (b)~checking whether a formula is a
  locally closed separator is in NC (so, simpler than unreachability,
  which is P-complete).

  We further consider the more general problem of (existential)
  set-to-set reachability, where two sets of markings are given as
  convex polytopes. We show that, while our approach does not extend
  directly, we can efficiently certify unreachability via an
  altered Petri net.
\end{abstract}

\maketitle

\section{Introduction}
Petri nets form a widespread formalism of concurrency with several
applications ranging from the verification of concurrent programs to
the analysis of chemical systems. The reachability problem --- which
asks whether a marking $\msrc$ can reach
another marking $\mtgt$ --- is fundamental as a plethora of problems,
such as verifying safety properties, reduce to it
(e.g.\ \cite{GS92,FMWDR17,BMTZ21}).

Leroux has shown that unreachability in Petri nets can be witnessed by a Presburger \emph{separator}, i.e.,\ if a marking $\msrc$ cannot reach a marking $\mtgt$, then there exists a formula $\varphi$ of Presburger arithmetic such that: $\varphi(\msrc)$ holds; $\varphi$ is forward invariant, i.e.,  $\varphi(\vec{m})$ and $\vec{m} \rightarrow \vec{m}'$ imply $\varphi(\vec{m}'$); and $\varphi(\mtgt)$ does not hold~\cite{Leroux12}. Intuitively, $\varphi$ ``separates'' $\mtgt$ from the set of markings reachable from $\msrc$. Leroux's result leads to a very simple algorithm to decide the Petri net reachability problem, consisting of two semi-algorithms; the first one explores the markings reachable from $\msrc$, and halts if and when it hits $\mtgt$, while the second enumerates formulas from Presburger arithmetic, and halts if and when it hits a separator. 
 
Separators can be used as \emph{explanations} and as formal \emph{certificates}. Verifying a safety property can be reduced to proving that a target marking (or set of markings) is not reachable from a source marking, and a separator is an invariant of the system that \emph{explains} why the property holds. Further, if a reachability tool produces separators, then the user can check that the properties of a separator indeed hold, and so trust the result even if they do not trust the tool (e.g., because it has not been verified, or is executed on a remote faster machine). Yet, in order to be useful as explanations and certificates, separators have to satisfy two requirements: (1)~they should not be too large, and (2)~checking that a formula is a separator should have low complexity, and, in particular, strictly lower complexity than deciding reachability. This does not hold, at least in the worst-case, for the separators of~\cite{Leroux12}: In the worst case, the separator has a size at least Ackermannian in the Petri net size (a consequence of the fact that the reachability problem is Ackermann-complete~\cite{LS19,Ler21,CO21}) and the complexity of the check is at least exponential.


In this paper, we show that, unlike the above, \emph{continuous} Petri nets do have separators satisfying properties~(1) and~(2). Continuous Petri nets are a relaxation of the standard Petri net model, called \emph{discrete} in the following, in which transitions are allowed to fire ``fluidly'': instead of firing once, consuming $i_p$ tokens from each input place $p$ and adding $o_q$ tokens to each output place $q$, a transition can fire $\alpha$ times for any nonnegative real number $\alpha$, consuming and adding $\alpha \cdot i_p$ and $\alpha \cdot o_q$ tokens, respectively. Continuous Petri nets are interesting in their own right~\cite{DavidA10}, and moreover as an overapproximation of the discrete model. In particular, if $\mtgt$ is not reachable from $\msrc$ under the continuous semantics, then it is also not under the discrete one. As reachability in continuous Petri nets is P-complete~\cite{FH15}, and so drastically more tractable than discrete reachability, this approximation is used in many tools for the verification of discrete Petri nets, VAS, or multiset rewriting systems (e.g.~\cite{BFHH17,BlondinEH0M20,EsparzaHJM20}).

It is easy to see that unreachability in continuous Petri nets can be witnessed by separators expressible in linear arithmetic (the first-order theory of the reals with addition and order). Indeed, Blondin \textit{et al}.\ show in~\cite{BFHH17} that the continuous reachability relation is expressible by an existential formula $\textit{reach}(\vec{m}, \vec{m}')$ of linear arithmetic, from which we can obtain a separator for any pair of unreachable markings. Namely, for all markings $\msrc$ and $\mtgt$, if $\mtgt$ is not reachable from $\msrc$, then the formula $\textit{sep}_{\msrc}(\vec{m}) \defeq \neg \textit{reach}(\msrc, \vec{m})$ is a separator. Further, $\textit{reach}(\vec{m}, \vec{m}')$ has only linear size. However, these separators do not satisfy property (2). Indeed, while the reachability problem for continuous Petri nets is P-complete~\cite{FH15}, checking if a formula of linear arithmetic is a separator is coNP-hard, even for quantifier-free formulas in disjunctive normal form, a very small fragment. So, the separators arising from~\cite{BFHH17} cannot be directly used as certificates.

In this paper, we overcome this problem. We identify a class of \emph{locally closed separators}, satisfying the following properties: unreachability can always be witnessed by locally closed separators; locally closed separators can be constructed in polynomial time; and checking whether a formula is a locally closed separator is computationally strictly\footnote{Assuming $\text{P} \neq \text{NC}$.} easier than deciding unreachability. Let us examine the last claim in more detail. While the reachability problem for continuous Petri nets is decidable in polynomial time, it is still time consuming for larger models, which can have tens of thousands of nodes. Indeed, for a Petri net with $n$ places and $m$ transitions, the algorithm of~\cite{FH15} requires to solve $\bigO(m^2)$ linear programming problems in $n$ variables, each of them with up to $m$ constraints.  Moreover, since the problem is P-complete, it is unlikely that a parallel computer can significantly improve performance. We prove that, on the contrary, checking if a formula is a locally closed separator is in NC rather than P-complete, and so efficiently parallelizable. Further, the checking algorithm only requires to solve linear programming problems in \emph{a single} variable. 

We further consider ``set-to-set reachability'' where one must
determine whether there exists a marking $\msrc \in A$ that can reach
some marking $\mtgt \in B$. This naturally generalizes the case where
$A = \{\msrc\}$ and $B = \{\mtgt\}$. We focus on the case where sets
$A$ and $B$ are convex polytopes, as it can be solved in polynomial
time for continuous Petri nets~\cite{BH17}. We prove that,
unfortunately, we cannot validate locally closed separators in NC for
set-to-set reachability. Nonetheless, we show that we can efficiently
construct and validate locally closed separators (respectively in P
and NC), at the cost of working with an altered Petri net that encodes
$A$ and $B$.

The paper is organized as follows. Section~\ref{sec:prelims}
introduces terminology, and defines separators (actually, a slightly
different notion called bi-separators). Section~\ref{sec:informal}
recalls the characterization of the reachability relation given by
Fraca and Haddad in~\cite{FH15}, and derives a characterization of
\emph{un}reachability suitable for finding
bi-separators. Section~\ref{sec:certificates} shows that checking the
separators derivable from~\cite{BFHH17} is coNP-hard, and introduces
locally closed bi-separators. Sections~\ref{sec:construct}
and~\ref{sec:verif} show that locally closed bi-separators satisfy the
aforementioned properties~(1) and~(2). Finally,
Section~\ref{sec:set2set} shows the limitation and extension of our
approach to set-to-set reachability.

\section{Preliminaries}\label{sec:prelims}
\paragraph*{Numbers, vectors and relations.}

We write $\N$, $\R$ and $\Rpos$ to denote the naturals (including
$0$), reals, and non-negative reals (including $0$). Given $a, b \in
\N$, we write $[a..b]$ to denote $\{a, a+1, \ldots, b\}$. Let $S$ be a
finite set. We write $\vec{e}_s$ to denote the unit vector $\vec{e}_s
\in \R^S$ such that $\vec{e}_s(s) = 1$ and $\vec{e}_s(t) = 0$ for all
$s, t \in S$ such that $t \neq s$. Given $\vec{x}, \vec{y} \in \R^S$,
we write $\vec{x} \sim_S \vec{y}$ to indicate that $\vec{x}(s) \sim
\vec{y}(s)$ for all $s \in S$, where $\sim$ is a total order such as
$\leq$. We define the \emph{support} of a vector $\vec{x} \in \R^S$ as
$\supp{\vec{x}} \defeq \{s \in S : \vec{x}(s) > 0\}$. We write
$\vec{x}(S) \defeq \sum_{s \in S} \vec{x}(s)$. The \emph{transpose} of
a binary relation $\Rel$ is $\rev{\Rel} \defeq \{(y, x) : (x, y) \in
\Rel\}$.

\paragraph*{Petri nets.}

A \emph{Petri net}\footnote{In this work, ``Petri nets'' stands for
``continuous Petri nets''. In other words, we will consider standard
Petri nets, but equipped with a \emph{continuous} reachability
relation. We will work over the reals, but note that it is known that
working over the rationals is equivalent. For decidability issues, we
will assume input numbers to be rationals.} is a tuple $\pn = (P, T,
F)$ where $P$ and $T$ are disjoint finite sets, whose elements are
respectively called \emph{places} and \emph{transitions}, and where $F
= (\mat{F}_-, \mat{F}_+)$ with $\mat{F}_-, \mat{F}_+ \colon P \times
T \to \N$. For every $t \in T$, vectors
$\prevec{t}, \postvec{t} \in \N^P$ are respectively defined as the
column of $\mat{F}_-$ and $\mat{F}_+$ associated to $t$, i.e.\
$\prevec{t} \defeq \mat{F}_- \cdot \vec{e}_t$ and $\postvec{t}
\defeq \mat{F}_+ \cdot \vec{e}_t$. A \emph{marking} is a vector
$\vec{m} \in \Rpos^P$. For every $Q \subseteq P$, let
$\vec{m}(Q) \defeq \sum_{p \in Q} \vec{m}(p)$. We say that transition $t$ is
\emph{$\alpha$-enabled} if $\vec{m} \geq \alpha \prevec{t}$ holds. If
this is the case, then $t$ can be \emph{$\alpha$-fired} from
$\vec{m}$, which leads to marking $\vec{m}' \defeq \vec{m} - \alpha
\prevec{t} + \alpha \postvec{t}$, which we denote $\vec{m}
\trans{\alpha t} \vec{m}'$. A transition is \emph{enabled}
if it is $\alpha$-enabled for some real number $\alpha > 0$. We define $\mat{F} \defeq \mat{F}_+ -
\mat{F}_-$ and $\effect{t} \defeq \mat{F} \cdot \vec{e}_t$. In
particular, $\vec{m} \trans{\alpha t} \vec{m}'$ implies $\vec{m}' =
\vec{m} + \alpha \effect{t}$. For example, for the Petri net
of Figure~\ref{fig:pn}:
\[
\{p_1 \mapsto 2, p_2 \mapsto 0, p_3 \mapsto 0, p_4 \mapsto 0\}
\trans{(1/2) t_1} \{p_1 \mapsto 3/2, p_2 \mapsto 1/2, p_3 \mapsto 0,
p_4 \mapsto 0\}.
\]
Moreover, w.r.t.\ to orderings $p_1 < \cdots < p_4$ (rows) and $t_1 <
\cdots < t_4$ (columns):
\[
\mat{F}_{-} =
\begin{bmatrix}
   1 & 2 & 2 & 0 \\
   0 & 0 & 1 & 0 \\
   0 & 0 & 0 & 1 \\
   0 & 1 & 0 & 0
\end{bmatrix},\
\mat{F}_{+} =
\begin{bmatrix}
   0 & 0 & 1 & 0 \\
   1 & 0 & 0 & 0 \\
   0 & 1 & 1 & 0 \\
   0 & 1 & 0 & 1
\end{bmatrix}\
\text{ and }\
\mat{F} =
\begin{bmatrix}
  -1 & -2 & -1 &  0 \\
   1 &  0 & -1 &  0 \\
   0 &  1 &  1 & -1 \\
   0 &  0 &  0 &  1
\end{bmatrix}.
\]

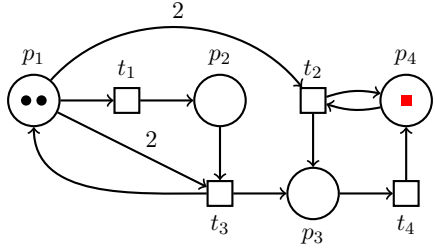
\begin{figure}[tbh]
  \centering
  \vspace*{-15pt}
  \begin{tikzpicture}[node distance=1.5cm, auto, thick, transform shape, scale=0.825]
    \node[place,      label=above:$p_1$,  tokens=2]   (p1) {};
    \node[transition, label=above:$t_1$, right of=p1] (t1) {};

    \node[place,      label=above:$p_2$, right of=t1] (p4) {};
    \node[transition, label=below:$t_3$, below of=p4] (t3) {};

    \node[transition, label=above:$t_2$, right of=p4] (t2) {};
    \node[place,      label=above:$p_4$, right of=t2] (p2) {}
         [children are tokens]
         child {node[token, red, rectangle] {}};

    \node[place,      label=below:$p_3$, below of=t2] (p3) {}
       ;

    \node[transition, label=below:$t_4$, right of=p3] (t4) {};

    \path[->]
    (p1) edge node {} (t1)
    (t1) edge node {} (p4)
    
    (p1) edge[bend left=50] node {$2$} (t2)
    (t2) edge node {} (p3)
    (t2) edge[bend left=15] node {} (p2)
    (p2) edge[bend left=15] node {} (t2)

    (p1) edge node[xshift=2pt, yshift=-3pt] {$2$} (t3)
    (t3) edge[out=180, in=-90]  node {}    (p1)
    (p4) edge                   node {}    (t3)
    (t3) edge                   node {}    (p3)

    (p3) edge node {} (t4)
    (t4) edge node {} (p2)
    ;
  \end{tikzpicture}
  \caption{A Petri net and two markings $\msrc = \{p_1 \mapsto 2, p_2
    \mapsto 0, p_3 \mapsto 0, p_4 \mapsto 0\}$ (black circles) and
    $\mtgt = \{p_1 \mapsto 0, p_2 \mapsto 0, p_3 \mapsto 0, p_4
    \mapsto 1\}$ (colored squares).}\label{fig:pn}
\end{figure}

A sequence $\sigma = \alpha_1 t_1 \cdots \alpha_n t_n$ is
a \emph{firing sequence} from $\msrc$ to $\mtgt$ if there are markings
$\vec{m}_0, \ldots, \vec{m}_n$ satisfying $\msrc
= \vec{m}_0 \trans{\alpha_1 t_1} \vec{m}_1 \cdots \trans{\alpha_n
t_n} \vec{m}_n = \mtgt$. We write
$\vec{m}_0 \trans{\sigma} \vec{m}_n$. We say that
$\msrc$ \emph{enables} $\sigma$, and that $\mtgt$ enables $\sigma$
backwards, or \emph{backward-enables} $\sigma$. The \emph{support} of
$\sigma$ is the set $\{t_1, \ldots, t_n\}$. For example, for the Petri
net of Figure~\ref{fig:pn}, we have $\msrc \trans{\sigma} \mtgt$
where
\begin{align*}
  \msrc &= \{p_1 \mapsto 2, p_2 \mapsto 0, p_3 \mapsto 0, p_4 \mapsto 0\}, \\
  \mtgt &= \{p_1 \mapsto 0, p_2 \mapsto 0, p_3 \mapsto 0, p_4 \mapsto 1\}, \\
  \sigma &= (1/2) t_1\ (1/2) t_3\ (1/2) t_4\ (1/2) t_2 \ (1/2) t_4.
\end{align*}

We write $\vec{m} \Utrans{} \vec{m}'$ to denote
that $\vec{m} \trans{\alpha t} \vec{m}'$ for some transition $t$ and some $\alpha > 0$, and
$\vec{m} \Utrans{*} \vec{m}'$ to denote that $\vec{m} \trans{\sigma} \vec{m}'$ for some firinng
sequence $\sigma$.

The Petri net $\pn_U$ is obtained by removing transitions
$T \setminus U$ from $\pn$. In particular,
$\vec{m} \Utrans{U^*} \vec{m}'$ holds in $\pn$ iff
$\vec{m} \Utrans{*} \vec{m}'$ holds in $\pn_U$.


The \emph{transpose} of $\pn = (P, T, (\mat{F}_-, \mat{F}_+))$ is
$\rev{\pn} \defeq (P, T, (\mat{F}_+, \mat{F}_-))$. 
We have 
$\msrc \trans{\sigma} \mtgt$ in $\pn$ if{}f $\mtgt\trans{\tau} \msrc$ in
$\rev{\pn}$, where $\tau$ is the reverse of $\sigma$. 
For $U \subseteq
T$, we write $U^\transpose$ to denote $U$ in the context of
$\rev{\pn}$. 

\paragraph*{Linear arithmetic and Farkas' lemma.}

An \emph{atomic proposition} is a linear inequality of the form
$\vec{a} \vec{x} \leq b$ or $\vec{a} \vec{x} < b$, where $b$ and the
components of $\vec{a}$ are over $\R$. Such a proposition is
\emph{homogeneous} if $b = 0$. A \emph{linear formula} is a
first-order formula over atomic propositions with variables ranging
over $\Rpos$ (the classical definition uses $\R$, but in our context
variables will encode markings.) The \emph{solutions} of a linear
formula $\varphi$, denoted $\sol{\varphi}$, are the assignments to the
free variables of $\varphi$ that satisfy $\varphi$. A linear formula
is \emph{homogeneous} if all of its atomic propositions are
homogeneous. For every formula $\varphi(\vec{x}, \vec{y})$ where
$\vec{x}$ and $\vec{y}$ have the same arity, we write $\rev{\varphi}$
to denote the formula that syntactically swaps $\vec{x}$ and
$\vec{y}$, so that $\sol{\rev{\varphi}} =
\rev{\sol{\varphi}}$. Throughout the paper, we will use Farkas' lemma,
a fundamental result of linear arithmetic that rephrases the absence
of solution to a system into the existence of one for another
system. The lemma has many variants, of which we give just one (e.g.\
see~\cite[Proposition~6.4.3(iii)]{GM07}):

\begin{lem}[Farkas' lemma]\label{lem:farkas}
  Let $\mat{A} \in \R^{m \times n}$ and $\vec{b} \in \R^m$. The
  formula $\mat{A}\vec{x} \leq \vec{b}$ has no solution iff
  $\mat{A}^\transpose \vec{y}
  = \vec{0} \land \vec{b}^\transpose \vec{y} <
  0 \land \vec{y} \geq \vec{0}$ has a solution.
\end{lem}

%
%

We give some geometric intuition for this lemma. We consider only the case $n=m$. Let $\vec{y}$ be a vector satisfying  $\vec{y} \geq \vec{0}$ (that is, $\vec{y}$ belongs to the positive orthant) and $\vec{b}^\transpose \vec{y} < 0$. The hyperplane $H$ perpendicular to $\vec{y}$ divides the space into two half-spaces. Since $\vec{b}^\transpose \vec{y} < 0$, the vectors $\vec{y}$ and $\vec{b}$ lie in opposite half-spaces; we say that $H$ \emph{separates} $\vec{b}$ from $\vec{y}$. Further, since $\vec{y}$ belongs to the positive orthant, the complete positive orthant lies in the same half-space as $\vec{y}$. In other words,  $H$ separates $\vec{b}$ not only from $\vec{y}$, but from the  positive orthant.

We argue that $\mat{A}^\transpose \vec{y} = \vec{0}$ holds for some $\vec{y}$ satisfying  $\vec{y} \geq \vec{0}$ and $\vec{b}^\transpose \vec{y} < 0$ if{}f $\mat{A}\vec{x} \leq \vec{b}$ has no solution. It is the case that  $\mat{A}^\transpose \vec{y} = \vec{0}$ holds for a vector $\vec{y}$ if{}f every column vector of $\mat{A}$ is perpendicular to $\vec{y}$, and so if{}f  every column vector of $\mat{A}$ lies in $H$. Since $H$ separates $\vec{y}$ and $\vec{b}$, and $\mat{A}\vec{x}$ is a linear combination of the columns of $\mat{A}$ for every vector $\vec{x}$, every column vector of $\mat{A}$ lies in $H$ if{}f $\vec{b} - \mat{A}\vec{x}$ lies on the same half-space as $\vec{b}$ for every $\vec{x}$. Since $H$ separates $\vec{b}$ from the positive orthant, this is the case if{}f  $\vec{b} - \mat{A}\vec{x}$ is not in the positive orthant for any $\vec{x}$, and so if{}f $\mat{A}\vec{x} \leq \vec{b}$ has no solution.

\subsection{Separators and bi-separators}

Let us fix a Petri net $\pn = (P, T, F)$ and two markings
$\msrc, \mtgt \in \Rpos^P$.

\begin{defi}\label{def:separator}
  A \emph{separator} for $(\msrc,\mtgt)$ is a linear formula $\varphi$
  over $\Rpos^P$ such that:
  \begin{enumerate}
  \item $\msrc \in \sol{\varphi}$;

  \item $\varphi$ is \emph{forward invariant}, i.e., $\vec{m} \in
    \sol{\varphi}$ and $\vec{m} \trans{} \vec{m'}$ implies $\vec{m}'
    \in \sol{\varphi}$; and

  \item $\mtgt \notin \sol{\varphi}$.
  \end{enumerate}
\end{defi}

It follows immediately from the definition that if there is a
separator $\varphi$ for $(\msrc, \mtgt)$, then $\msrc \not\Utrans{*}
\mtgt$. Thus, in order to show that $\msrc \not\Utrans{*} \mtgt$ in
$\pn$, we can either give a separator for $(\msrc, \mtgt)$
w.r.t.\ $\pn$, or a separator for $(\mtgt, \msrc)$
w.r.t.\ $\rev{\pn}$. Let us call them \emph{forward} and
\emph{backward} separators. Loosely speaking, a forward separator
shows that $\mtgt$ is not among the markings reachable from $\msrc$,
and a backward separator shows that $\msrc$ is not among the markings
backward-reachable from $\mtgt$. Bi-separators are formulas from which
we can easily obtain forward and backward separators. The symmetry
w.r.t.\ forward and backward reachability make them easier to handle.



\begin{defi}\label{def:bisep}
  A linear formula $\varphi$ over $\Rpos^P \times \Rpos^P$ is
  \begin{itemize}
  \item \emph{forward invariant} if $(\vec{m}, \vec{m}') \in
    \sol{\varphi}$ and $\vec{m}' \trans{} \vec{m}''$ imply $(\vec{m},
    \vec{m}'') \in \sol{\varphi}$;

  \item \emph{backward invariant} if $(\vec{m}', \vec{m}'') \in
    \sol{\varphi}$ and $\vec{m} \trans{} \vec{m}'$ imply $(\vec{m},
    \vec{m}'') \in \sol{\varphi}$; and

  \item \emph{bi-invariant} if it is forward and backward invariant.
  \end{itemize}
  A \emph{bi-separator} for $(\msrc, \mtgt)$ is a bi-invariant linear
  formula $\varphi$ such that $(\msrc, \msrc) \in \sol{\varphi}$,
  $(\mtgt, \mtgt) \in \sol{\varphi}$ and $(\msrc, \mtgt) \notin
  \sol{\varphi}$.
\end{defi}

The following proposition shows how to obtain separators from
bi-separators.

\begin{prop}
  Let $\varphi$ be a bi-separator for $(\msrc, \mtgt)$. The following
  holds:
  \begin{itemize}
  \item $\mathmakebox[27pt][l]{\psi(\vec{m})} \defeq \varphi(\msrc,
    \vec{m})$ is a separator for $(\msrc, \mtgt)$ in $\pn$;

  \item $\mathmakebox[27pt][l]{\psi'(\vec{m})} \defeq \varphi(\vec{m},
    \mtgt)$ is a separator for $(\mtgt,\msrc)$ in $\rev{\pn}$.
  \end{itemize}
\end{prop}

\begin{proof}
  It suffices to prove the first statement, the second is
  symmetric. We have $\msrc \in \sol{\psi}$ and $\mtgt \notin
  \sol{\psi}$ as $(\msrc, \msrc) \in \sol{\varphi}$ and $(\msrc,
  \mtgt) \notin \sol{\varphi}$.

  It remains to show that $\psi$ is forward invariant. Let $\vec{m}
  \in \sol{\psi}$ and $\vec{m} \trans{\alpha t} \vec{m}'$. Since
  $(\msrc, \vec{m}) \in \sol{\varphi}$ and $\varphi$ is forward
  invariant, it is the case that $(\msrc, \vec{m}') \in
  \sol{\varphi}$. Hence, $\vec{m}' \in \sol{\psi}$ as desired.
\end{proof}

\section{A characterization of unreachability}\label{sec:informal}
Given a Petri net $\pn = (P, T, F)$, a set of transitions $U \subseteq T$, and two markings
  $\msrc, \mtgt \in \Rpos^P$, we write $\msrc \etrans{U} \mtgt$ to denote
that $\msrc \trans{\sigma} \mtgt$ for some sequence $\sigma$ with support $U$. 
In words,  $\msrc \trans{\sigma} \mtgt$ denotes that  $\mtgt$ can be reached from $\msrc$ by firing
 \emph{all} transitions of $U$ and \emph{only} transitions of $U$. We say that $\mtgt$ is \emph{$U$-reachable}
from $\msrc$

In~\cite{FH15}, Fraca and Haddad gave the following characterization
of $U$-reachability in continuous Petri nets:

\begin{thmC}[\cite{FH15}]\label{thm:reach:charac}
  Let $\pn = (P, T, F)$ be a Petri net, let $U \subseteq T$, and let
  $\msrc, \mtgt \in \Rpos^P$. It is the case that $\msrc \etrans{U}
  \mtgt$ if{}f the following
  conditions hold:
  \begin{enumerate} 
  \item some vector $\vec{x} \in \Rpos^T$ with support $U$ satisfies
    $\msrc + \mat{F} \vec{x} = \mtgt$,
    
  \item some firing sequence $\sigma$ with support $U$ is enabled at
    $\msrc$, and
    
  \item some firing sequence $\tau$ with support $U$ is
    backward-enabled at $\mtgt$.
  \end{enumerate}
  Furthermore, these conditions can be checked in polynomial time.
\end{thmC}



Theorem~\ref{thm:reach:charac} provides a witness of $U$-reachability in the shape of a vector $\vec{x}$ and two firing sequences 
$\sigma$ and $\tau$. In this section, we find a witness of $U$-unreachability, i.e., a witness showing that $\msrc \not\etrans{U}
  \mtgt$. 

Consider the logical form of Theorem~\ref{thm:reach:charac}:

\begin{equation}
\label{logic}
\begin{gathered}
 \msrc \etrans{U} \mtgt   \\ 
\iff  (\exists \vec{x} \in \Rpos^T
 \colon A_1(\vec{x})) \land (\exists \sigma \in T^* \colon A_2(\sigma))
\land (\exists \tau \in T^* \colon A_3(\tau)),
\end{gathered}
\end{equation}

\noindent where $A_1(\vec{x}), A_2(\sigma), A_3(\tau)$ denote that $\vec{x}$ has support $U$ and satisfies
    $\msrc + \mat{F} \vec{x} = \mtgt$, that $\sigma$ has support $U$ and is enabled at
    $\msrc$, and that $\tau$ has support $U$ and is
    backward-enabled at $\mtgt$, respectively.  The theorem is logically equivalent to

\begin{equation}
\label{logicequiv}
\begin{gathered}
\msrc \not\etrans{U} \mtgt \\ 
\iff  (\forall \vec{x} \in \Rpos^T \colon \neg A_1(\vec{x}))
\lor (\forall \sigma \in T^* \colon \neg A_2(\sigma)) \lor (\forall \tau \in T^*
\colon \neg A_3(\tau)).
\end{gathered}
\end{equation}

\noindent To obtain a witness of $U$-unreachability, we
proceed in two steps:

\begin{itemize}
\item In Section~\ref{sec:B2B3}, we define predicates $B_2(Q)$ and $B_3(R)$, where $Q,R \subseteq P$, satisfying

\begin{equation}
\label{logicequiv2}
\begin{gathered}
(\forall \sigma \in T^* \colon \neg A_2(\sigma))  \Longleftrightarrow (\exists Q \subseteq P \colon B_2(Q)), \\
(\forall \tau \in T^* \colon \neg A_3(\tau)) \Longleftrightarrow (\exists R \subseteq P \colon B_3(R)). \\
\end{gathered}
\end{equation}

\item In Section~\ref{sec:B1}, we define a predicate $B_1(\vec{y})$, where $\vec{y} \in \Rpos^P$, satisfying

\begin{equation}
\label{logicequiv3}
\begin{gathered}
(\forall \vec{x} \in \Rpos^T \colon \neg A_1(\vec{x})) \Longrightarrow (\exists \vec{y} \in \Rpos^P \colon B_1(\vec{y})) \Longrightarrow \msrc \not\etrans{U} \mtgt.
\end{gathered}
\end{equation}
\end{itemize}

Here, a remark is in order. Observe that we do not claim that
``$\forall \vec{x} \in \Rpos^T \colon \neg A_1(\vec{x})$'' is
equivalent to ``$\exists \vec{y} \in \Rpos^P \colon B_1(\vec{y})$.''
Rather, we claim that the former implies the latter, and the latter
implies $U$-unreachability.

Altogether, applying propositional logic to  (\ref{logicequiv}), (\ref{logicequiv2}), and (\ref{logicequiv3}) we obtain:
\begin{equation}
\label{logicb}
\begin{gathered}
\msrc \not\etrans{U} \mtgt \\ 
\iff  (\exists \vec{y} \in \Rpos^P \colon B_1(\vec{y}))
\vee (\exists Q \subseteq P \colon B_2(Q)) \vee (\exists R \subseteq P
\colon B_3(R)).
\end{gathered}
\end{equation}

\noindent This shows that $U$-unreachability is witnessed either by a vector $\vec{y}$, a set $Q$, or a set $R$.

\subsection{The predicates $B_2(Q)$ and $B_3(R)$}
\label{sec:B2B3}

The predicates $B_2(Q)$ and $B_3(R)$ were already implicitly defined in~\cite{FH15} in terms of sets of places called \emph{siphons}
and \emph{traps}.  Given a set of places $X$, let
$\pre{X}$ (resp.\ $\post{X}$) be the set of transitions $t$ such that
$\mat{F}_{+}(p,t) > 0$ (resp.\ $\mat{F}_{-}(p,t) > 0$) for some $p \in
X$. A \emph{siphon} of $\pn$ is a subset $Q$ of places such that
$\pre{Q} \subseteq \post{Q}$. A \emph{trap} is a subset $R$ of places
such that $\post{R} \subseteq \pre{R}$. Informally, empty siphons
remain empty, and marked traps remain marked. Formally, if $\vec{m}
\trans{} \vec{m}'$, then $\vec{m}(Q) = 0$ implies $\vec{m}'(Q) = 0$,
and $\vec{m}(R) > 0$ implies $\vec{m}'(R) > 0$. Fraca and Haddad proved in \cite{FH15}:

\begin{propC}[\cite{FH15}]\label{prop:siphon-trap}
  Let $\pn = (P, T, F)$ be a Petri net, let $U \subseteq T$, and let
  $\vec{m} \in \Rpos^P$. The following statements hold:
  \begin{itemize}
  \item No firing sequence with support $U$ is enabled at $\vec{m}$
    if{}f there exists a siphon $Q$ of $\pn_U$ such that $\post{Q} \neq
    \emptyset$ and  $\vec{m}(Q) = 0$;
    
  \item No firing sequence with support $U$ is backward-enabled at
    $\vec{m}$ if{}f there exists a trap $R$ of $\pn_U$ such that $\pre{R} \neq
    \emptyset$ and $\vec{m}(R) = 0$.
  \end{itemize}
\end{propC}
\noindent The if-direction of
the proposition is easy to prove. A siphon $Q$ of $\pn_U$ satisfies
$\post{Q} \subseteq U$. Since $Q$ is empty at $\vec{m}$, if we only
fire transitions from $U$ then $Q$ remains empty, and so no transition
of $\post{Q}$ ever becomes enabled. So, transitions of $\post{Q}$ can
only fire after transitions that do not belong to $U$ have fired
first. But no such firing sequence has support $U$, and we are
done. The case of traps is analogous. For the only-if direction, we
refer the reader to~\cite{FH15}.

By Proposition \ref{prop:siphon-trap}, the predicates
\begin{align*}
B_2(Q) &\defeq \text{$Q$ is a siphon of  $\pn_U$ such that $\post{Q} \neq \emptyset$ and $\msrc(Q) = 0$},  \\
B_3(R) &\defeq \text{$R$ is a trap of  $\pn_U$ such that $\pre{R} \neq \emptyset$ and $ \mtgt(R) = 0$},
\end{align*}
\noindent satisfy (\ref{logicequiv2}).

\subsection{The predicate $B_1(\mathbf{y})$} 
\label{sec:B1}

 In order to define the predicate $B_1(\vec{y})$, we need to
 introduce \emph{exclusion functions}.

\begin{defi}\label{def:mark:eq:nosol}
Let $\pn = (P, T, F)$ be a Petri net, let $\msrc, \mtgt \in \Rpos^P$
  and let $S \subseteq S' \subseteq T$. An \emph{exclusion function
    for $(S, S')$} is a function $f \colon \Rpos^P \to \R$ such that
  \begin{enumerate}
  \item for all $t \in S'$, if $\vec{m} \trans{t} \vec{m}'$ then $f(\vec{m}) \leq
    f(\vec{m}')$;  and
    
  \item either $f(\msrc) > f(\mtgt)$, or $f(\msrc) = f(\mtgt)$ and
    there exists $t \in S$ such that if $\vec{m} \trans{t} \vec{m}'$
    then $f(\vec{m}) < f(\vec{m}')$.
  \end{enumerate}

\noindent  An \emph{exclusion function for $S$} is an exclusion function for $(S,S)$. An exclusion function $f \colon \Rpos^P \to \R$ is \emph{linear} if there exists a vector $\vec{y} \in \R^P$ such that $f(\vec{m}) = \sum_{p \in P} \vec{y}(p) \cdot \vec{m}(p)$ for every $\vec{m} \in \Rpos^P$. Abusing language, we speak of the linear exclusion function $\vec{y}$.
\end{defi}

We define:
\[
B_1(\vec{y}) \defeq \text{$\vec{y}$ is a linear exclusion function for $U$}.
\]
We can easily show the second implication of (\ref{logicequiv3}), namely that an exclusion function for $U$ is a witness of $U$-unreachabilty:

\begin{lem}
If there exists a (linear) exclusion function for $U \subseteq T$, then $\msrc \not\etrans{U} \mtgt$.
\end{lem}
\begin{proof}
Let $f$ be an exclusion function for $U$.
Call $f(\vec{m})$ the \emph{value} of $\vec{m}$. By the definition of an exclusion function,
either $\mtgt$ has lower value than $\msrc$ and no transition of $U$
decreases the value, or $\msrc$ and $\mtgt$ have the same value, no
transition of $U$ decreases the value, and at least one transition of $U$ increases it. So, it
is impossible to reach $\mtgt$ from $\msrc$ by firing \emph{all} and
\emph{only} the transitions of $U$. 
\end{proof}


The first  implication of  (\ref{logicequiv3}) is shown in the forthcoming Proposition~\ref{prop:mark:eq:nosol}.
We first need a technical lemma, one of the many consequences of Farkas' lemma. Recall that, given ${\sim} \in \{\leq,
=, \geq\}$, the notation $\vec{x} \sim_U \vec{y}$ stands for
``$\vec{x}(t) \sim \vec{y}(t)$ for all $t \in U$''.

\begin{lem}\label{prop:farkas:supp}
  The system $\exists \vec{x} \geq \vec{0} : \mat{A} \vec{x} = \vec{b}
  \land U \subseteq \supp{\vec{x}} \subseteq V$ has no solution iff
  this system has a solution: $\exists \vec{y} : \mat{A}^\transpose
  \vec{y} \geq_{V} \vec{0} \land \vec{b}^\transpose\vec{y} \leq 0
  \land \vec{b}^\transpose\vec{y} < \sum_{t \in U} (\mat{A}^\transpose
  \vec{y})(t)$.
\end{lem}

Rather than giving a proof (which appears in the appendix), let us give
some geometric intuition. Consider the case in which $\emptyset =
U \subseteq V = T$. In this case, $U\subseteq \supp{\vec{x}} \subseteq
V$ holds vacuously, and we can replace ``$\geq_{V}$'' by ``$\geq$''
and ``$\vec{b}^\transpose\vec{y} < \sum_{t \in U}
(\mat{A}^\transpose \vec{y})(t)$'' by ``$\vec{b}^\transpose\vec{y} <
0$''. The conditions can now be geometrically interpreted as follows:
\begin{enumerate}
\item $\exists \vec{x} \geq \vec{0} : \mat{A} \vec{x} = \vec{b}$. The vector $\vec{b}$ lies in the cone spanned by the columns of $\mat{A}$.

\item $\exists \vec{y} : \mat{A}^\transpose \vec{y} \geq \vec{0} \land \vec{b}^\transpose\vec{y} < 0$. There exists a hyperplane $H$ (the hyperplane perpendicular to $\vec{y}$) such that all columns of 
$\mat{A}$ lie on the same side of the hyperplane as $\vec{y}$ (because $\mat{A}^\transpose \vec{y} \geq \vec{0}$),
while $\vec{b}$ lies on the opposite side (because $\vec{b}^\transpose\vec{y} <
0$). 
\end{enumerate}
Clearly, the two conditions are incompatible. The conditions of the general case $U \subseteq V$ are similar, but restricted to the columns of $\mat{A}$ 
for which $\vec{x}$ is positive: (1)~now states that $\vec{b}$ lies in the cone spanned by these columns, and (2)~that some hyperplane separates 
these columns from $\vec{b}$.

We are now ready to prove the first implication of (\ref{logicequiv3}), namely 
\[(\forall \vec{x} \in \Rpos^T \colon \neg A_1(\vec{x})) \Longrightarrow (\exists \vec{y} \in \Rpos^P \colon \neg B_1(\vec{y})).\]
Recall that $A_1(\vec{x})$ means that $\vec{x}$ has support $U$ and satisfies $\msrc + \mat{F} \vec{x} = \mtgt$, and 
$B_1(\vec{y})$ means that $\vec{y}$ is a linear exclusion function for $U$. We prove a slightly more general result; the first implication
of (\ref{logicequiv3}) is the special case $S = U = S'$:

\begin{prop}\label{prop:mark:eq:nosol}
   Let $\pn = (P, T, F)$ be a Petri net, let $\msrc, \mtgt \in
  \Rpos^P$, and let $S \subseteq S' \subseteq T$. If no vector $\vec{x}
  \in \Rpos^T$ satisfies $S \subseteq \supp{\vec{x}} \subseteq S'$ and
  $\msrc + \mat{F} \vec{x} = \mtgt$, then there exists a linear
  exclusion function for $(S,S')$.
\end{prop}

\begin{proof}

  Assume no such $\vec{x} \in \Rpos^T$ exists. Let
  $\vec{b} \defeq \mtgt - \msrc$. By
  Lemma~\ref{prop:farkas:supp}, there exists $\vec{y} \in \R^P$
  such that: $ \mat{F}^\transpose \vec{y} \geq_{S'} \vec{0} \land
  \vec{b}^\transpose\vec{y} \leq 0 \land \vec{b}^\transpose\vec{y} <
  \sum_{s \in S} (\mat{F}^\transpose \vec{y})_s$. We show that
  $f(\vec{k}) \defeq \vec{y}^\transpose \vec{k}$ is a linear exclusion
  function for $(S, S')$.

  \begin{enumerate}
  \item We have
    $
    f(\mtgt) - f(\msrc)
    =
    \vec{y}^\transpose\mtgt - \vec{y}^\transpose\msrc
    =
    \vec{y}^\transpose(\mtgt - \msrc)
    =
    \vec{y}^\transpose \vec{b}
    =
    \vec{b}^\transpose \vec{y}
    \leq
    0,
    $
    and hence $f(\mtgt) \leq f(\msrc)$.\medskip
    
  \item Let $\vec{m} \trans{\lambda s} \vec{m}'$ with $s \in S'$ and
    $\lambda \in \Rpos$. We have $\vec{m}' = \vec{m} + \lambda \mat{F}
    \vec{e}_s$. Thus:    
    $
    f(\vec{m}')
    = \vec{y}^\transpose \vec{m}'
    %
    %
    %
    = \vec{y}^\transpose \vec{m} + \lambda (\vec{y}^\transpose
    \mat{F}) \vec{e}_s
    = \vec{y}^\transpose \vec{m} + \lambda (\mat{F}^\transpose
    \vec{y})^\transpose \vec{e}_s
    \geq \vec{y}^\transpose \vec{m}
    = f(\vec{m})$,
    where the inequality follows from $\lambda > 0$,
    $\mat{F}^\transpose \vec{y} \geq_{S'} \vec{0}$ and
    $s \in S'$.\medskip

  \item Recall that $\vec{b}^\transpose \vec{y} \leq 0$ and $\sum_{s
    \in S} (\mat{F}^\transpose \vec{y})_s > \vec{b}^\transpose
    \vec{y}$. If the latter sum equals zero, then $\vec{b}^\transpose
    \vec{y} < 0$, and hence we are done since $f(\mtgt) - f(\msrc) =
    \vec{b}^\transpose\vec{y} < 0$.

    Otherwise, we have $\sum_{s \in S} (\mat{F}^\transpose \vec{y})_s
    > 0$ since $S \subseteq S'$ and $\mat{F}^\transpose \vec{y}
    \geq_{S'} \vec{0}$. Therefore, there exists a transition $s \in S$
    such that $(\mat{F}^\transpose \vec{y})_s > 0$. Let $\vec{m}
    \trans{s} \vec{m}'$. We have $\vec{m}' = \vec{m} + \lambda \mat{F}
    \vec{e}_s$ for some $\lambda > 0$. Thus,
    $f(\vec{m}')
    = \vec{y}^\transpose \vec{m} + \lambda
      (\mat{F}^\transpose \vec{y})^\transpose \vec{e}_s
    > \vec{y}^\transpose \vec{m}
    = f(\vec{m})$, where the inequality holds by $\lambda > 0$ and
    $(\mat{F}^\transpose \vec{y})_s > 0$.\qedhere
  \end{enumerate}  
\end{proof}

Putting together Theorem~\ref{thm:reach:charac} with Proposition~\ref{prop:siphon-trap}, and Proposition~\ref{prop:mark:eq:nosol},
we obtain the following characterization of unreachability:

\begin{prop}\label{prop:unreach:charac}
  Let $\pn = (P, T, F)$ be a Petri net, let $U \subseteq T$, and 
  $\msrc, \mtgt \in \Rpos^P$. It is the case that
  $\msrc \not\etrans{U} \mtgt$ if{}f
  \begin{enumerate} 
  \item there exists a linear exclusion function for $U$, or
  \label{itm:unreach:charac:eq}
  
  \item there exists a siphon $Q$ of $\pn_U$ such that
    $\post{Q} \neq \emptyset$ and $\msrc(Q) = 0$, or
    \label{itm:unreach:charac:fwd}

  \item there exists a trap $R$ of $\pn_U$ such that $\pre{R} \neq
    \emptyset$ and $\mtgt(R) = 0$.\label{itm:reach:uncharac:bwd}
  \end{enumerate}
\end{prop}

This proposition shows that we can produce a
witness of unreachability for a given support either as an exclusion function, a siphon, or a
trap. The following example shows how to use this to show that a marking
cannot be reached from another one.

\begin{exa}\label{ex:witnesses2}
  Consider the Petri net of Figure~\ref{fig:pn}, but with $\mtgt
  \defeq \{p_1 \mapsto 0, p_2 \mapsto 0, p_3 \mapsto 1, p_4 \mapsto
  0\}$ as target. We prove $\msrc \not\Utrans{*} \mtgt$.  For the sake
  of contradiction, assume $\msrc \etrans{U} \mtgt$ for some $U
  \subseteq T$. We proceed in several steps:
  \begin{itemize}
  \item \emph{Claim: $t_4 \notin U$}. The function $f(\vec{m}) \defeq
    \vec{m}(p_4)$ is an exclusion function for $T$. Indeed, since no
    transition decreases the number of tokens of $p_4$,
    $\vec{m}\trans{t} \vec{m}'$ implies $f(\vec{m}) \leq f(\vec{m}')$
    for every transition $t \in T$. Furthermore, $f(\msrc) = 0 =
    f(\mtgt)$, and, since $t_4$ adds tokens to $p_4$, $\vec{m}
    \trans{t_4} \vec{m}'$ implies $f(\vec{m}) < f(\vec{m}')$. It
    follows that no firing sequence from $\msrc$ to $\mtgt$ can fire
    $t_4$. Therefore, $\mtgt$ is reachable from $\msrc$ in $\pn$ if{}f it is reachable from $\msrc$ in 
    $\pn_{T \setminus\{t_4\}}$.
        
    \smallskip

  \item \emph{Claim: $t_2 \notin U$}. The set $Q \defeq \{p_4\}$ is a
    siphon of $\pn_{T \setminus \{t_4\}}$ (but not of $\pn$). Since
    $\msrc(Q) = 0$, it is impossible to use transitions of $\pn_{T
      \setminus\{t_4\}}$ that consume from $Q$, i.e.\ transitions of
    $\post{Q} = \{t_2\}$.  So $\mtgt$ is reachable from $\msrc$ in $\pn$ if{}f it is reachable from $\msrc$ in 
    $\pn_{T \setminus\{t_2, t_4\}}$
    
    \smallskip

  \item \emph{Claim: $t_1, t_3 \notin U$}. The set $R \defeq \{p_1,
    p_2\}$ is a trap of $\pn_{T \setminus \{t_2, t_4\}}$ (but not of
    $\pn_{T \setminus \{t_4\}}$). Since $\mtgt(R) = 0$, it is
    impossible to reach $\mtgt$ using transitions of $\pn_{T \setminus
      \{t_2, t_4\}}$ that produce in $R$, i.e.\ transitions of
    $\pre{R} = \{t_1, t_3\}$.
  \end{itemize}
  By the claims, $\mtgt$ is reachable from $\msrc$ in $\pn$  if{}f it is reachable from $\msrc$ in $\pn_{\emptyset}$. But this can only happen
  if $\msrc = \mtgt$, which is not the case.
\end{exa}

In the next section, we transform this collection of witnesses of unreachability into one single separator that can be used as
certificate.

\section{Separators as certificates}\label{sec:certificates}
Let $\pn = (P, T, F)$ be a Petri net and let $\msrc, \mtgt \in
\Rpos^P$ be markings of $\pn$. From~\cite{BFHH17}, one can easily
show that if $\msrc \not\Utrans{*} \mtgt$, then there is a separator
for $(\msrc,\mtgt)$. Indeed, \cite[Prop.~3.2]{BFHH17} shows that there
exists an existential formula $\psi$ of linear arithmetic such that
$\vec{m} \Utrans{*} \vec{m}'$ if{}f $(\vec{m}, \vec{m}') \in
\sol{\psi}$. Thus, the formula $\varphi(\vec{m}) \defeq \psi(\msrc,
\vec{m})$ is a separator.

However, $\varphi$ is not adequate as a \emph{certificate} of
unreachability. Indeed, checking a certificate for
$\msrc \not\Utrans{*} \mtgt$ should have smaller complexity than
deciding whether $\msrc \Utrans{*} \mtgt$. This is not the case for
existential linear formulas, because $\msrc \Utrans{*} \mtgt$ can be
decided in polynomial time, but checking that an existential linear
formula is a separator is coNP-hard.

\begin{prop}\label{prop:sep:coNP}
  The problem of determining whether an existential linear formula
  $\varphi$ is a separator for $(\msrc, \mtgt)$ is coNP-hard, even if
  $\varphi$ is a quantifier-free formula in DNF and homogeneous (i.e.,
  each atomic proposition $\vec{a} \vec{x} \sim b$ is such that $b =
  0$.)
\end{prop}

\begin{proof}
  We give a reduction from the problem of determining whether a DNF
  boolean formula $\psi(x_1, \ldots, x_m)$ is a tautology. Since
  $\psi$ is in DNF, we can test whether it is satisfiable in
  polynomial time. Obviously, if it is not, then it is not a
  tautology. Thus, we consider the case where $\psi$ is
  satisfiable. We define $\psi'$ as $\psi$ but where literals are
  modified as follows: $x_i$ becomes $\vec{m}(x_i) > 0$, and $\neg
  x_i$ becomes $\vec{m}(x_i) \leq 0$. We construct a Petri net $\pn =
  (P, T, F)$, where $P = \{x_1, \ldots, x_m, p, q\}$, $T = \{t\}$, and
  $t$ produces a token in $p$. Note that all places but $p$ are
  disconnected. Let $\mtgt \defeq \vec{0}$ and $\msrc$ be such that
  $\msrc(p) = \msrc(q) = 1$ and $(\msrc(x_1), \ldots, \msrc(x_n))$
  encodes an (arbitrary) assignment that satisfies $\psi$. We claim
  that the following formula $\varphi$ is a separator for $(\msrc,
  \mtgt)$ iff $\psi$ is a tautology:
  \[
  \varphi(\vec{m}) \defeq (\vec{m}(p) \leq 0 \land \vec{m}(q) > 0)
  \lor (\vec{m}(p) > 0 \land \vec{m}(q) > 0 \land \psi'(\vec{m})).
  \]
  
  As $\msrc \in \sol{\varphi}$ and $\mtgt \notin \sol{\varphi}$, it
  suffices to show that $\varphi$ is forward invariant if{}f $\psi$ is
  a tautology

  $\Rightarrow$) Let $(y_1, \ldots, y_m) \in \{0, 1\}^n$. We show that
  $\psi(y_1, \ldots, y_m)$ holds. Let $\vec{m}$ be such that
  $\vec{m}(p) = 0$, $\vec{m}(q) = 1$ and $\vec{m}(x_i) = y_i$ for
  every $i \in [1..m]$. We have $\vec{m} \in \sol{\varphi}$ since the
  first disjunct is satisfied. We can fire $t$ in $\vec{m}$ which
  leads to a marking $\vec{m}'$ where $\vec{m}'(p) > 0$, $\vec{m}'(q)
  = 1$ and $\vec{m}'(x_i) = y_i$ for every $i \in [1..m]$. Since
  $\varphi$ is forward invariant by assumption, we have $\vec{m}' \in
  \sol{\varphi}$. Since $\vec{m}'(p) > 0$ and $\vec{m}'(q) > 0$, this
  means that the second disjunct is satisfied. In particular, this
  means that $\psi'(\vec{m}')$ holds, and hence that $\psi(y_1,
  \ldots, y_m)$ holds.

  $\Leftarrow$) Let $\vec{m} \in \sol{\varphi}$ and $\vec{m}
  \trans{\alpha t} \vec{m}'$. We have $\vec{m}(q) > 0$. By definition
  of $t$, we have $\vec{m}'(p) > 0$, and $\vec{m}'$ equal to $\vec{m}$
  on all other places. Since $\psi$ is a tautology, $\psi'(\vec{m})$
  holds. Thus, $\vec{m}' \in \sol{\varphi}$ holds as the second
  disjunct is satisfied.
\end{proof}

\begin{rem}
 The problem of determining whether a quantifier-free linear formula
 $\varphi$ is a separator for $(\msrc, \mtgt)$ is in coNP. Indeed, we
 can specify that a formula $\varphi$ separates $(\msrc, \mtgt)$ as
 follows:
 \[
   \varphi(\msrc) \land \neg \varphi(\mtgt) \land
   \forall \vec{m} \geq \vec{0}\ \forall \alpha > 0
   \bigwedge_{t \in T} [\varphi(\vec{m}) \land \vec{m}
   \geq \alpha \prevec{t}) \rightarrow \varphi(\vec{m} +
   \alpha \effect{t})].
 \]
 Membership follows since the universal fragment of linear arithmetic
 is in coNP~\cite{Son85}.
\end{rem}

In the rest of the section, we introduce locally closed bi-separators,
and then, in Sections~\ref{sec:construct} and~\ref{sec:verif}, we
respectively prove that they satisfy the following:

\begin{itemize}
\item If $\msrc \not\Utrans{*} \mtgt$, then some locally closed
  bi-separator for $(\msrc, \mtgt)$ can be computed in polynomial time
  by solving at most $\bigO(|T|)$ linear programs;

\item Deciding whether a formula is a locally closed bi-separator is in NC.
\end{itemize}

\subsection{Locally closed bi-separators}

The most difficult part of checking that a formula $\varphi$ is a
bi-separator consists of checking that it is forward and backward
invariant. Let us focus on forward invariance, backward invariance
being symmetric.

Recall the definition: for all markings $\vec{m}, \vec{m'}, \vec{m''}$
and every transition $t$: if $(\vec{m}, \vec{m}') \in \sol{\varphi}$
and $\vec{m}' \trans{\alpha t} \vec{m}''$ then $(\vec{m}, \vec{m}'')
\in \sol{\varphi}$. Assume now that $\varphi$ is in DNF, i.e., a
disjunction of clauses $\varphi = \varphi_1 \lor \cdots \lor
\varphi_n$. The forward invariance check can be decomposed into $n$
smaller checks, one for each $i \in [1..n]$, of the form: if
$(\vec{m}, \vec{m}') \in \sol{\varphi_i}$, then $(\vec{m}, \vec{m}'')
\in \sol{\varphi}$. However, in general the check \emph{cannot} be
decomposed into \emph{local} checks of the form: there exists $j \in
[1..m]$ such that $(\vec{m}, \vec{m}') \in \sol{\varphi_i}$ implies
$(\vec{m}, \vec{m}'') \in \sol{\varphi_j}$. Indeed, while this
property is sufficient for forward invariance, it is not
necessary. Intuitively, locally closed bi-separators are separators
where invariance can be established by local checks.

For the formal definition, we need to introduce some notations. 
Given a transition $t$ and atomic propositions $\psi, \psi'$, we say
that $\psi$ \emph{$t$-implies} $\psi'$, written $\psi \leadsto_t
\psi'$, if $(\vec{m}, \vec{m}') \in \sol{\psi}$ and $\vec{m}'
\trans{\alpha t} \vec{m}''$ implies $(\vec{m}, \vec{m}'') \in
\sol{\psi'}$. We further say that a clause $\psi = \psi_1 \land \cdots
\land \psi_m$ \emph{$t$-implies} a clause $\psi' = \psi_1' \land
\cdots \land \psi_n'$, written $\psi \leadsto_t \psi'$, if for every
$j \in [1..n]$, there exists $i \in [1..m]$ such that $\psi_i
\leadsto_t \psi_j'$.
        
\begin{defi}
  A linear formula $\varphi$ is \emph{locally closed} w.r.t.\ $\pn =
  (P, T, F)$ if:
  \begin{itemize}
  \item $\varphi = \varphi_1 \lor \cdots \lor \varphi_n$ is
    quantifier-free, in DNF and homogeneous,
    
  \item for every $t \in T$ and every $i \in [1..n]$, there exists $j
    \in [1..n]$ s.t.\ $\varphi_i \leadsto_t \varphi_j$,
    
  \item for every $t \in \rev{T}$ and every $i \in [1..n]$, there
    exists $j \in [1..n]$ s.t.\ $\rev{\varphi_i} \leadsto_t
    \rev{\varphi_j}$.
  \end{itemize}
\end{defi}


Note that the definition is semantic. We make the straightforward but
crucial observation that:

\begin{prop}\label{prop:bi-invariance}
  Locally closed formulas are bi-invariant.
\end{prop}

\begin{proof}
  Let $\varphi = \varphi_1 \lor \cdots \lor \varphi_n$ be a locally
  closed formula. We only consider the forward case; the other case is
  symmetric. Let $(\vec{m}, \vec{m}') \in \sol{\varphi}$ and $\vec{m}'
  \trans{\alpha t} \vec{m}''$. Let $i \in [1..n]$ be such that
  $(\vec{m}, \vec{m}') \in \sol{\varphi_i}$. Since $\varphi$ is
  locally closed, there exists $j \in [1..n]$ such that $\varphi_i
  \leadsto_t \varphi_j$. For every atomic proposition $\psi'$ of
  $\varphi_j$, there exists an atomic proposition $\psi$ of
  $\varphi_i$ such that $\psi \leadsto_t \psi'$. Since each atomic
  proposition of $\varphi_i$ is satisfied by $(\vec{m}, \vec{m}')$, we
  obtain $(\vec{m}, \vec{m}'') \in \sol{\varphi_j}$.
\end{proof}

Proposition~\ref{prop:bi-invariance} justifies the following
definition:

\begin{defi}
  A \emph{locally closed bi-separator} for $(\msrc, \mtgt)$ is a
  locally closed formula $\varphi$ s.t.\ $(\msrc, \msrc) \in
  \sol{\varphi}$, $(\mtgt, \mtgt) \in \sol{\varphi}$ and $(\msrc,
  \mtgt) \notin \sol{\varphi}$.
\end{defi}

Indeed, by Proposition~\ref{prop:bi-invariance}, a locally closed
bi-separator is a bi-separator, as the bi-invariance condition of
Definition~\ref{def:bisep} follows from local closedness.

\section{Constructing locally closed bi-separators}\label{sec:construct}
In this section, we prove that unreachability can always be witnessed
by locally closed bi-separators of polynomial size and computable in
polynomial time. The proof uses the results of
Section~\ref{sec:informal}.

\begin{thm}\label{thm:unreach:certificate}
  If $\msrc \not\Utrans{U^*} \mtgt$, then there is a locally closed
  bi-separator $\varphi$ for $(\msrc,\mtgt)$ w.r.t.\ $\pn_U$. Further,
  $\varphi = \bigvee_{1 \leq i \leq n} \varphi_i$, where $n \leq 2|U|
  + 1$ and each $\varphi_i$ contains at most $2|U| + 1$ atomic
  propositions. Moreover, $\varphi$ is computable in polynomial time.
\end{thm}

\begin{proof}
  We proceed by induction on $|U|$. First consider $U =
  \emptyset$. There must exist $p \in P$ such that $\msrc(p) \neq
  \mtgt(p)$. Take $\varphi(\vec{m}, \vec{m}') \defeq \vec{e}_p \vec{m}
  \leq \vec{e}_p \vec{m}'$ or $-\vec{e}_p \vec{m} \leq -\vec{e}_p
  \vec{m}'$, depending on whether $\msrc(p) > \mtgt(p)$ or $\msrc(p) <
  \mtgt(p)$.

  Now, assume that $U \neq \emptyset$. We make a case distinction on whether the system
  $$\exists
  \vec{x} \in \Rpos^T : \msrc + \mat{F} \vec{x} = \mtgt \land
  \supp{\vec{x}} \subseteq U \ .$$
\noindent has a solution. Suppose first that it has no solution. By
  Proposition~\ref{prop:mark:eq:nosol}, taking $S =\emptyset$ and $S'
  = U$, there is a linear exclusion function for $(\emptyset, U)$,
  i.e.\ a linear function $f$ satisfying (1) $f(\msrc) > f(\mtgt)$, and (2)
  $\vec{m} \trans{u} \vec{m}'$ implies $f(\vec{m}) \leq
    f(\vec{m}')$ for all $u \in U$.
  (The first item holds due to Item~2 of
  Definition~\ref{def:mark:eq:nosol} and $S = \emptyset$.) So we can
  take $\varphi(\vec{m}, \vec{m}') \defeq (f(\vec{m}) \leq
  f(\vec{m}'))$.

  Suppose now that the system has a solution $\vec{x} \in \Rpos^U$. By
  convexity, we can suppose that $\supp{\vec{x}} \subseteq U$ is
  maximal. Indeed, if $\vec{x}'$ and $\vec{x}''$ are solutions, then
  $(1/2)\vec{x}' + (1/2)\vec{x}''$ is a solution with support
  $\supp{\vec{x}'} \cup \supp{\vec{x}''}$. Let $U' \defeq
  \supp{\vec{x}}$. For every $t \in U \setminus U'$, consider the
  system of Proposition~\ref{prop:mark:eq:nosol} with $S = \{t\}$ and
  $S' = U$. By maximality of $U' \subseteq U$, none of these systems
  has a solution. Consequently, for each $t \in U \setminus U'$,
  Proposition~\ref{prop:mark:eq:nosol} yields a linear exclusion
  function for $(\{t\}, U)$, i.e.\ a linear function $f_t$ that
  satisfies:
  \begin{enumerate}
    \setcounter{enumi}{2}
    
  \item $f_t(\msrc) \geq f_t(\mtgt)$, 
    
  \item $\vec{m} \trans{u} \vec{m}'$ implies $f_t(\vec{m}) \leq
    f_t(\vec{m}')$ for all $u \in U$,

  \item either $f_t(\msrc) > f_t(\mtgt)$, or $\vec{m} \trans{t}
    \vec{m}'$ implies $f_t(\vec{m}) < f_t(\vec{m}')$.
  \end{enumerate}

  If $f_t(\msrc) > f_t(\mtgt)$ holds for some $t \in U \setminus U'$,
  then we are done by taking $\varphi(\vec{m}, \vec{m}') \defeq
  (f_t(\vec{m}) \leq f_t(\vec{m}'))$ as Item~4 ensures that $\varphi
  \leadsto_u \varphi$ for every $u \in U$. So assume it does not hold for any $t \in U \setminus U'$, i.e. assume that \ $f_t(\msrc) = f_t(\mtgt)$ holds, and the second
  disjunct of Item~5 holds for all $t \in U \setminus U'$. This is the most involved case.
  Let
  \[
  \varphi_\text{inv}(\vec{m}, \vec{m}') \defeq
  \bigwedge_{\mathclap{t \in U \setminus U'}} (f_t(\vec{m}) \leq
  f_t(\vec{m}'))\
  \text{ and }\
  \varphi_t(\vec{m}, \vec{m'}) \defeq  
    (f_t(\vec{m}) < f_t(\vec{m}')).
  \]
  Let $Q, R \subseteq P$ be respectively the maximal siphon and trap
  of $\pn_{U'}$ such that $\msrc(Q) = 0$ and $\mtgt(R) = 0$
  (well-defined by closure under union). Let $U'' \defeq U' \setminus
  (\post{Q} \cup \pre{R})$. By Theorem~\ref{thm:reach:charac} and
  Proposition~\ref{prop:siphon-trap}, $\post{Q} \cup \pre{R} \neq
  \emptyset$. Thus, $U''$ is a strict subset of $U'$, and, by
  induction hypothesis, there is a locally closed bi-separator
  w.r.t.\ $\pn_{U''}$ of the form $\psi = \bigvee_{1 \leq i \leq m}
  \psi_i$ that satisfies the claim for set $U''$. Let
  \begin{multline*}
    \varphi(\vec{m}, \vec{m}') \defeq
    \bigvee_{\mathclap{t \in U \setminus U'}}
    \varphi_t(\vec{m}, \vec{m}') \lor
    [\varphi_\text{inv}(\vec{m}, \vec{m}') \land \vec{m}(Q) +
      \vec{m}'(R) > 0] \lor {} \\[-12pt]
    \bigvee_{\mathclap{1 \leq i \leq m}} [\varphi_\text{inv}(\vec{m},
      \vec{m}') \land \vec{m}(R) + \vec{m}'(Q) \leq 0 \land
      \psi_i(\vec{m}, \vec{m}')].
  \end{multline*}
  As $(\msrc, \msrc) \in \sol{\varphi_\text{inv}}$ and $(\msrc, \msrc)
  \in \sol{\psi}$, we have $(\msrc, \msrc) \in
  \sol{\varphi}$. Similarly, $(\mtgt, \mtgt) \in \sol{\varphi}$. By
  Item~3, $(\msrc, \mtgt) \notin \sol{\bigvee_{t \in U \setminus U'}
    \varphi_t(\vec{m}, \vec{m}')}$. Further, $\msrc(Q) + \mtgt(R) = 0$
  and $(\msrc, \mtgt) \notin \sol{\psi}$. So, $(\msrc, \mtgt) \notin
  \sol{\varphi}$.

  \smallskip The number of disjuncts of $\varphi$ is $|U \setminus U'|
  + 1 + m$ and hence at most
  \begin{align*}
    |U \setminus U'| + 1 + 2|U''| + 1
    &\leq
    (|U| - |U''|) + 1 + 2|U''| + 1
    && \text{(since $U \supseteq U' \supseteq U''$)} \\
    &=
    |U| + |U''| + 2 \\
    &\leq
    |U| + (|U| - 1) + 2
    && \text{(since $U'' \subset U$)} \\
    &=
    2|U| + 1.
  \end{align*}
  The same bounds holds for the number of atomic propositions per
  disjunct.
  
  \smallskip  It remains to show that $\varphi(\vec{m}, \vec{m}')$ is locally closed
  w.r.t.\ $\pn_U$. We only consider the forward case, as the backward
  case is symmetric. Let $(\vec{m}, \vec{m}') \in \sol{\varphi}$ and
  $\vec{m}' \trans{u} \vec{m}''$ for some $u \in U$. By Item~4,
  $\varphi_t \leadsto_u \varphi_t$ holds for each $\varphi_t$. Indeed,
  $f_t(\vec{m}) < f_t(\vec{m}')$ and $\vec{m}' \trans{u} \vec{m}''$
  imply $f_t(\vec{m}) < f_t(\vec{m}') \leq f_t(\vec{m}'')$, and hence
  $f_t(\vec{m}) < f_t(\vec{m}'')$. To handle the other clauses, we
  make a case distinction on $u$.

  \begin{itemize}
  \item \emph{Case $u \in U \setminus U'$.} Atomic proposition $\theta
    = (f_u(\vec{m}) \leq f_u(\vec{m}'))$ of $\varphi_\text{inv}$
    satisfies $\theta \leadsto_u \varphi_u$. Indeed, if $f_u(\vec{m})
    \leq f_u(\vec{m}')$ and $\vec{m}' \trans{u} \vec{m}''$, then we
    have $f_u(\vec{m}) < f_u(\vec{m}')$ by Item~5.
    \smallskip

  \item \emph{Case $u \in U'$.} By Item~4, each atomic proposition
    $\theta$ of $\varphi_\text{inv}$ satisfies $\theta \leadsto_u
    \theta$.
    \smallskip
    \begin{itemize}
    \item \emph{Case $u \in \pre{R}$.} We have $\theta' \leadsto_u
      (\vec{m}(Q) + \vec{m}'(R) > 0)$ for any atomic proposition
      $\theta'$, since $\vec{m}' \trans{u} \vec{m}''$ implies
      $\vec{m}''(R) > 0$ (regardless of $\theta'$).
    \smallskip

    \item \emph{Case $u \in \post{Q}$.} If $\vec{m}'(Q) \leq 0$, then
      $u$ is disabled in $\vec{m}'$. Thus, it only remains to handle
      $\theta_{> 0} \defeq (\vec{m}(Q) + \vec{m}'(R) > 0)$. Since $R$
      is a trap of $\pn_{U'}$, firing $u$ from $\vec{m}'$ does not
      empty $R$, and hence $\theta_{> 0} \leadsto_u \theta_{> 0}$.
    \smallskip

  \item \emph{Case $u \in U''$.} Let $\theta_{\leq 0} \defeq
    (\vec{m}(R) + \vec{m}'(Q) \leq 0)$ and $\theta_{> 0} \defeq
    (\vec{m}(Q) + \vec{m}'(R) > 0)$. Since $Q$ and $R$ are
    respectively a siphon and trap of $\pn_{U'}$, we have
    $\theta_{\leq 0} \leadsto_u \theta_{\leq 0}$ and $\theta_{> 0}
    \leadsto_u \theta_{> 0}$. Moreover, by induction hypothesis, for
    every $i \in [1..m]$, there exists $j \in [1..m]$ such that
    $\psi_i \leadsto_u \psi_j$.
    \end{itemize}
  \end{itemize}
  
  We conclude the proof by observing that it is constructive and can
  be turned into Algorithm~\ref{alg:certificate}. The procedure works
  in polynomial time. Indeed, there are at most $|U|$ recursive
  calls. Moreover, each set can be obtained in polynomial time via
  either linear programming or maximal siphons/traps
  computations~\cite{desel1995free}.
\end{proof}

\begin{algorithm}[h!]
  \DontPrintSemicolon
  \SetKwFunction{algcert}{bi-separator}
  \SetKwProg{myproc}{}{}{}
  
  \KwIn{$\pn = (P, T, F)$, $U \subseteq T$ and 
        $\msrc, \mtgt \in \Qpos^P$ s.t.\ $\msrc \not\Utrans{U^*} \mtgt$}
  \KwOut{A locally closed bi-separator w.r.t.\ $\pn_U$}

  \algskip
  \myproc{\algcert{$U$}}{
    \If{$U = \emptyset$}{
      \textbf{pick} $p \in P \text{ such that } \msrc(p) \neq \mtgt(p)$\;
            
      \Return{$(\vec{a} \vec{m} \leq \vec{a} \vec{m}')$ where $\vec{a}
        \defeq \mathrm{sign}(\msrc(p) - \mtgt(p)) \cdot \vec{e}_p$}
    }
    \Else{
      $\mathmakebox[12pt][l]{\vec{b}} \defeq \mtgt - \msrc$\;
      
      $\mathmakebox[12pt][l]{X} \defeq \{\vec{x} \in \Rpos^T :
      \mat{F}\vec{x} = \vec{b}, \supp{\vec{x}} \subseteq U\}$\;
      
      $\mathmakebox[12pt][l]{Y_S} \defeq \{\vec{y} \in \R^P :
      \mat{F}^\transpose \vec{y} \geq_U \vec{0},
      \vec{b}^\transpose\vec{y} \leq 0, \vec{b}^\transpose\vec{y} <
      \sum_{s \in S} (\mat{F}^\transpose \vec{y})_s\}$\;
      \algskip
      
      \If{$X = \emptyset$}{
        \textbf{pick} $\vec{y} \in Y_\emptyset$ and       
        \Return{$(\vec{y}^\transpose \vec{m} \leq \vec{y}^\transpose
          \vec{m}')$}      
      }
      \Else{
        $U' \defeq \{u \in U : \vec{x}(u) > 0 \text{ for some }
        \vec{x} \in X\}$\;

        \algskip
        \For{$t \in U \setminus U'$}{
          \textbf{pick} $\vec{y}_t \in Y_{\{t\}}$;
          $f_t(\vec{m}) \defeq \vec{y}_t^\transpose \vec{m}$\;
          
          \lIf{$f_t(\msrc) > f_t(\mtgt)$}{
            \Return{$(f_t(\vec{m}) < f_t(\vec{m}'))$}
          }
        }
        
        \algskip
        $Q \defeq$ largest \makebox[30pt][l]{siphon} of $\pn_{U'}$
        such that $\msrc(Q) = 0$\;
        
        $R \defeq$ largest \makebox[30pt][l]{trap} of $\pn_{U'}$
        such that $\mtgt(R) = 0$\;
        
        $\varphi_\text{inv} \defeq
        \bigwedge_{t \in U \setminus U'} (f_t(\vec{m}) \leq
        f_t(\vec{m}'))$\;
        
        \algskip
        $\psi_1 \lor \cdots \lor \psi_m \defeq \algcert{$U'
          \setminus (\post{Q} \cup \pre{R})$}$\;
        
        \algskip
        \Return{} $\bigvee_{t \in U \setminus U'} \varphi_t(\vec{m},
          \vec{m}') \lor [\varphi_\text{inv}(\vec{m}, \vec{m}') \land
            \vec{m}(Q) + \vec{m}'(R) > 0] \lor {}$\;
          \hspace*{1.25cm}%
          $\bigvee_{1 \leq i \leq m} [\varphi_\text{inv}(\vec{m},
          \vec{m}') \land \vec{m}(R) + \vec{m}'(Q) \leq 0 \land
          \psi_i(\vec{m}, \vec{m}')]$\;
      }
    }
  }
  \algskip
  \caption{Construction of a locally closed bi-sep.\ for $(\msrc, \mtgt)$.}
  \label{alg:certificate}
\end{algorithm}

\begin{exa}
  Let us apply the construction of
  Theorem~\ref{thm:unreach:certificate} to the Petri net and the
  markings of Example~\ref{ex:witnesses2}: $\msrc = \{p_1 \mapsto 2,
  p_2 \mapsto 0, p_3 \mapsto 0, p_4 \mapsto 0\}$ and $\mtgt \defeq
  \{p_1 \mapsto 0, p_2 \mapsto 0, p_3 \mapsto 1, p_4 \mapsto 0\}$. The
  locally closed bi-separator is the formula $\varphi$ below, where
  the colored arrows represent the relations $\leadsto_{t_1}, \ldots,
  \leadsto_{t_4}$:
 
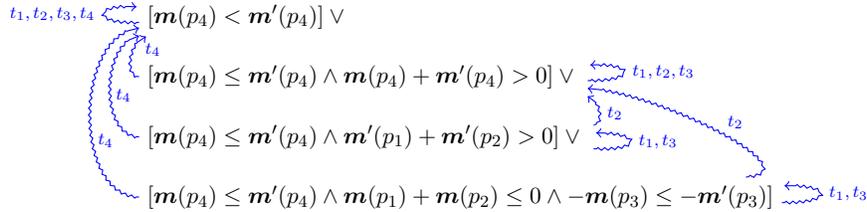
\begin{center}
   \begin{tikzpicture}[auto, node distance=1cm, transform shape, scale=0.8]
     \node[anchor=west] at (0, 0) (d1) {%
       $[\vec{m}(p_4) < \vec{m}'(p_4)] \lor {}$%
     };
     
     \node[anchor=west] at (0, -1) (d2) {%
       $[\vec{m}(p_4) \leq \vec{m}'(p_4) \land
         \vec{m}(p_4) + \vec{m}'(p_4) > 0] \lor {}$%
     };
     
     \node[anchor=west] at (0, -2) (d3) {%
       $[\vec{m}(p_4) \leq \vec{m}'(p_4) \land
         \vec{m}'(p_1) + \vec{m}'(p_2) > 0] \lor {}$%
     };
       
     \node[anchor=west] at (0, -3) (d4) {%
         $[\vec{m}(p_4) \leq \vec{m}'(p_4) \land \vec{m}(p_1) +
           \vec{m}(p_2) \leq 0 \land -\vec{m}(p_3) \leq -\vec{m}'(p_3)]$%
     };

     \path[->, blue, font=\scriptsize, decoration={zigzag,
     amplitude=0.5pt, segment length=0.9mm, pre length=1pt, post length=4pt}]
     (d1) edge[decorate, out=183, in=175,  loop] node {$t_1, t_2, t_3, t_4$} ()

     (d2) edge[decorate, out=180, in=191] node[xshift=18pt, yshift=-3pt] {$t_4$} (d1)
     (d3) edge[decorate, out=180, in=190] node[xshift=14pt, yshift=-3pt] {$t_4$} (d1)
     (d4) edge[decorate, out=180, in=186] node[xshift=15pt, yshift=-12pt] {$t_4$} (d1)

     (d2) edge[decorate, out=-1, in=3, loop] node[swap] {$t_1, t_2, t_3$} ()

     (d3) edge[decorate, out=-3, in=1, loop] node[swap] {$t_1, t_3$} () 
     (d3) edge[decorate, out=3,  in=-5]      node[swap, yshift=-8pt] {$t_2$} (d2) 

     (d4) edge[decorate, out=-1, in=2, loop] node[swap] {$t_1, t_3$} ()
     (d4) edge[decorate, out=4,  in=-3]      node[swap, yshift=-2pt] {$t_2$} (d2)
    ;
   \end{tikzpicture}
 \end{center}

  As an example, consider transition $t_4$ and these atomic
  propositions occurring within $\varphi$:
  \begin{align*}
    \theta(\vec{m}, \vec{m}')
    &\defeq \vec{m}(p_4) \leq \vec{m}'(p_4), \\
    \theta'(\vec{m}, \vec{m}')
    &\defeq \vec{m}(p_4) < \vec{m}'(p_4).
  \end{align*}
  Given $(\vec{m}, \vec{m}') \in \sol{\theta}$ and $\vec{m}'
  \trans{\alpha t_4} \vec{m}''$, we must have $\vec{m}'(p_4) <
  \vec{m}''(p_4)$ since $t_4$ produces in place $p_4$, and hence
  $(\vec{m}, \vec{m}'') \in \sol{\theta'}$. Thus, by definition of
  $t_4$-implication, we have $\theta \leadsto_{t_4} \theta'$. This
  explains why the first clause of $\varphi$ is $t_4$-implied by each
  clause of $\varphi$. Similar reasoning yields the other
  $t_i$-implications, which shows that $\varphi$ is bi-invariant.

  Recall that since $\varphi$ is a locally closed bi-separator, taking
  $\psi(\vec{m}) \defeq \varphi(\msrc, \vec{m})$ yields a forward
  separator. Since $\msrc(p_2) = \msrc(p_3) = \msrc(p_4) = 0$, several
  atomic propositions trivially hold. After making these
  simplifications, we obtain
  \begin{alignat*}{3}
    \psi(\vec{m})
    &\equiv\ &&
    \vec{m}(p_1) + \vec{m}(p_2) > 0 \lor \vec{m}(p_4) > 0.
  \end{alignat*}
  Similarly, we obtain this backward separator $\psi'(\vec{m}) \defeq
  \varphi(\vec{m}, \mtgt)$:
  \begin{alignat*}{3}
    \psi'(\vec{m})
    & \equiv\ &&
    \vec{m}(p_1) + \vec{m}(p_2) = 0 \land \vec{m}(p_3) \geq 1 \land
    \vec{m}(p_4) = 0.
  \end{alignat*}
  The backward separator $\psi'$ provides a much simpler proof of
  $\msrc \not\trans{*} \mtgt$ than the one of
  Example~\ref{ex:witnesses2}. The proof goes as follows: $\psi'$ is
  trivially backward invariant, because markings that only mark $p_3$
  do not backward-enable any transition. In particular, since $\mtgt$
  only marks $p_3$, it can only be reached from $\mtgt$.
\end{exa}

\section{Checking locally closed bi-separators is in NC}\label{sec:verif}
We show that the problem of deciding whether a given linear formula
is a locally closed bi-separator is in NC. To do so, we
provide a characterization of $\psi \leadsto_t \psi'$ for homogeneous
atomic propositions $\psi$ and $\psi'$. We only focus on forward
firability, as backward firability can be expressed as forward
firability in the transpose Petri net. Recall that
$\psi \leadsto_t \psi'$ holds iff the following holds:
\begin{align}
  (\vec{m}, \vec{m}') \in \sol{\psi} \text{ and }
  \vec{m}' \trans{\alpha t} \vec{m}'' \text{ imply }
  (\vec{m}, \vec{m}'') \in \sol{\psi'}.\tag{$*$}\label{eq:leadsto}
\end{align}
Property~\eqref{eq:leadsto} can be rephrased as:
\[
(\vec{m}, \vec{m}') \in \sol{\psi} \text{ and }
\vec{m}' \geq \alpha \cdot \prevec{t} \text{ imply }
(\vec{m}, \vec{m}' + \alpha \cdot \effect{t}) \in \sol{\psi'}.
\]
As we will see towards the end of the section, due to homogeneity of $\psi$ and $\psi'$, it
actually suffices to consider the case where $\alpha = 1$, which
yields this reformulation:
\[
\underbrace{\{(\vec{m}, \vec{m}') \in \sol{\psi} : \vec{m}' \geq
  \prevec{t}\}}_{X}
\subseteq
\underbrace{\{(\vec{m}, \vec{m}') : (\vec{m}, \vec{m}' + \effect{t})
  \in \sol{\psi'}\}}_{Y}.
\]
Therefore, testing $\psi \leadsto_t \psi'$ amounts to the inclusion
check $X
\subseteq Y$. Of course, if $X = \emptyset$, then this is
trivial. Hence, we will suppose that $X \neq \emptyset$, assuming for
now that it can somehow be tested efficiently. In the forthcoming
Propositions~\ref{prop:charac:incl} and~\ref{prop:incl:atomic}, we
will provide necessary and sufficient conditions for $X \subseteq Y$
to hold. In Proposition~\ref{prop:linprog:unary}, we will show that
these conditions are testable in NC. Then, in
Proposition~\ref{prop:charac:inv}, we will explain how to check
whether $X \neq \emptyset$ actually holds.

For $X \subseteq Y$, we can characterize the case of atomic
propositions $\psi$ that use ``$\leq$'' (rather than ``$<$''):

\begin{prop}\label{prop:charac:incl}
  Let $\vec{a}, \vec{a}', \vec{l} \in \R^{n}$ and $b' \in \R$. Let
  $X \defeq \{\vec{x} \in \R^n : \vec{a} \vec{x} \leq
  0 \land \vec{x} \geq \vec{l}\}$ and $Y \defeq \{\vec{x} \in \R^n
  : \vec{a}' \vec{x} \leq b'\}$ be such that $X \neq \emptyset$. It is
  the case that $X \subseteq Y$ iff there exists $\lambda \geq 0$ such
  that $\lambda \vec{a} \geq \vec{a}'$ and $-b' \leq (\lambda \vec{a}
  - \vec{a}') \vec{l}$.
\end{prop}

\begin{proof}
  The \emph{conical hull} of a finite set $V$ is defined as
  \[
  \conic{V} \defeq \left\{\sum_{\vec{v} \in V} \lambda_{\vec{v}} \vec{v}
  : \lambda_{\vec{v}} \in \Rpos\right\}.
  \]
  We first state a generalization of Farkas' lemma sometimes known as
  Haar's lemma (e.g.\ see~\cite[p.~216]{Sch86}):
  \begin{center}
    \begin{tabular}{|p{13.75cm}}
      Let $\mat{A} \in \R^{m \times n}$, $\vec{b} \in \R^m$, $\vec{a}'
      \in \R^n$ and $b' \in \R$ be such that $\sol{\mat{A}\vec{x} \leq
        \vec{b}} \neq \emptyset$. It is the case that
      $\sol{\mat{A}\vec{x} \leq \vec{b}} \subseteq
      \sol{\vec{a}'\vec{x} \leq b}$ iff $(\vec{a}', b') \in
      \conic{\{(\vec{0}, 1)\} \cup \{(\mat{A}_i, \vec{b}_i) : i \in
        [1..m]\}}$.
    \end{tabular}
  \end{center}
  We now prove the proposition. Let
  \[
  \mat{A} \defeq
  \begin{bmatrix}
    \vec{a} \\
    -\mat{I}
  \end{bmatrix}
  \text{ and }
  \vec{b} \defeq
  \begin{pmatrix}
    0 \\
    -\vec{l}
  \end{pmatrix}.
  \]
  Note that $X = \sol{\mat{A}\vec{x} \leq \vec{b}}$. So, by Haar's
  lemma, we have $X \subseteq Y$ iff $(\vec{a}', b') \in
  \conic{\{(\vec{0}, 1)\} \cup \{(\mat{A}_i, \vec{b}_i) : i \in
    [1..m]\}}$. Consequently, we have:
  \begin{align*}
    X \subseteq Y
    &\iff
    \exists \lambda, \lambda_1, \ldots, \lambda_n \geq 0 :
    \vec{a}' = \lambda \vec{a} - \sum_{i=1}^n \lambda_i \vec{e}_i
    \text{ and }
    b' \geq -\sum_{i=1}^n \lambda_i \vec{l}_i \\
    &\iff
    \exists \lambda \geq 0 :
    \lambda \vec{a} \geq \vec{a}'
    \text{ and }
    -b' \leq (\lambda \vec{a} - \vec{a}') \vec{l}.
    &&\qedhere
  \end{align*}
\end{proof}

We now give the conditions for all four combinations of ``$\leq$'' and
``$<$'':

\begin{restatable}{prop}{propinclatomic}\label{prop:incl:atomic}
  Let $\vec{a}, \vec{a}' \in \R^n$, $b' \in \R$, $\vec{l} \geq
  \vec{0}$ and ${\sim}, {\sim'} \in \{\leq, <\}$. Let $X_{\sim} \defeq
  \{\vec{x} \geq \vec{l} : \vec{a} \vec{x} \sim \vec{0}\}$ and
  $Y_{\sim'} \defeq \{\vec{x} \in \R^n : \vec{a}' \vec{x} \sim' b'\}$
  be such that $X_{\sim} \neq \emptyset$. It holds that $X_{\sim}
  \subseteq Y_{\sim'}$ iff there exists $\lambda \geq 0$
  s.t.\ $\lambda \vec{a} \geq \vec{a}'$ and one of the following
  holds:
  
  \begin{enumerate}
  \item ${\sim'} = {\leq}$ and $-b' \leq (\lambda \vec{a} - \vec{a}')
    \vec{l}$;

  \item ${\sim} = {\leq}$, ${\sim'} = {<}$, and $-b' < (\lambda
    \vec{a} - \vec{a}') \vec{l}$;

  \item ${\sim} = {<}$, ${\sim'} = {<}$, and either $-b' < (\lambda
    \vec{a} - \vec{a}') \vec{l}$ or $-b' = (\lambda \vec{a} -
    \vec{a}') \vec{l} \land \lambda > 0$.
  \end{enumerate}
\end{restatable}

\begin{proof}
  \leavevmode
  \begin{enumerate}
  \item If ${\sim} = {\leq}$, then it follows immediately from
    Proposition~\ref{prop:charac:incl}. Thus, assume ${\sim} =
    {<}$. We claim that $X_{<} \subseteq Y_{\leq}$ iff $X_{\leq}
    \subseteq Y_{\leq}$. The validity of this claim concludes the
    proof of this case as we have handled ${\sim} = {\leq}$ and as
    $X_{\leq} \supseteq X_{<} \neq \emptyset$.

    \medskip
    Let us show the claim. It is clear that $X_{<} \subseteq Y_{\leq}$
    is implied by $X_{\leq} \subseteq Y_{\leq}$. So, we only have to
    show direction from left to right. For the sake of contradiction,
    suppose that $X_{<} \subseteq Y_{\leq}$ and $X_{\leq}
    \not\subseteq Y_{\leq}$. Let $X_{=} \defeq X_{\leq} \setminus
    X_{<}$. Note that $X_{=} \neq \emptyset$. Let $\vec{x} \in X_{<}$
    and $\vec{x}' \in X_{=} \setminus Y_{\leq}$. We have $\vec{x},
    \vec{x}' \geq \vec{l}$, $\vec{a} \vec{x} < 0$, $\vec{a} \vec{x}' =
    0$, $\vec{a}' \vec{x} = c \leq b'$ and $\vec{a}' \vec{x}' = c' >
    b'$ for some $c, c' \in \R$. In particular, $b' \in [c, c')$. Let
      $\epsilon \in (0, 1]$ be such that $b' < \epsilon c + (1 -
    \epsilon) c'$. Let $\vec{x}'' \defeq \epsilon \vec{x} + (1 -
    \epsilon) \vec{x}'$. Observe that $\vec{x}'' \geq
    \vec{l}$. Moreover, we have:
    \begin{alignat*}{3}
      \vec{a} \vec{x}''
      &= \epsilon \vec{a} \vec{x} + (1 - \epsilon) \vec{a} \vec{x}'
      &&= \epsilon \vec{a} \vec{x}
      &&< 0, \\
      \vec{a}' \vec{x}''
      &= \epsilon \vec{a}' \vec{x} + (1 - \epsilon) \vec{a}' \vec{x}'
      &&= \epsilon c + (1 - \epsilon) c'
      &&> b'.      
    \end{alignat*}
    Therefore, we have $\vec{x}'' \in X_{<}$ and $\vec{x}'' \notin
    Y_{\leq}$, which is a contradiction.

    \medskip
  \item $\Rightarrow$) Since $X_{\leq} \subseteq Y_<$, the system
    $\exists \vec{x} : \vec{x} \geq \vec{l} \land \vec{a} \vec{x} \leq
    0 \land \vec{a}' \vec{x} \geq b'$ has no solution. In matrix
    notation, the system corresponds to $\exists \vec{x} : \mat{A}
    \vec{x} \leq \vec{c}$ where
    \[
    \mat{A} \defeq
    \begin{bmatrix}
      -\mat{I} \\
      \vec{a} \\
      -\vec{a}'
    \end{bmatrix}
    \text{ and }
    \vec{c} \defeq
    \begin{pmatrix}
      -\vec{l} \\
      0 \\
      -b'
    \end{pmatrix}.
    \]
    By Farkas' lemma (Lemma~\ref{lem:farkas}), $\mat{A}^\transpose
    \vec{y} = \vec{0}$ and $\vec{c}^\transpose \vec{y} < 0$ for some
    $\vec{y} \geq \vec{0}$. In other words,
    \[
    \exists \vec{z} \geq \vec{0}, \lambda, \lambda' \geq 0 :
    \lambda \vec{a} - \lambda' \vec{a}' = \vec{z} \land
    -\lambda' b' < \vec{z} \vec{l}.
    \]
    Since $\vec{z} \geq \vec{0}$, we have $\lambda \vec{a} \geq
    \lambda' \vec{a}' \land -\lambda' b' < (\lambda \vec{a} - \lambda'
    \vec{a}') \vec{l}$. If $\lambda' > 0$, then we are done by
    dividing all terms by $\lambda'$. For the sake of contradiction,
    suppose that $\lambda' = 0$. This means that $\lambda \vec{a} \geq
    \vec{0}$ and $0 < \lambda \vec{a} \vec{l}$. We necessarily have
    $\lambda > 0$ and $\vec{a} \vec{l} > 0$. Let $\vec{x} \in
    X_{\leq}$. We have $0 \geq \vec{a} \vec{x} \geq \vec{a} \vec{l} >
    0$, which is a contradiction.

    \medskip

    $\Leftarrow$) Let $\vec{x} \in X_{\leq}$. We have $\vec{a}'
    \vec{x} < b'$ and hence $\vec{x} \in Y_{<}$ as desired, since:
    \begin{align*}
      -b'
      &< (\lambda \vec{a} - \vec{a}')\vec{l} \\
      &\leq (\lambda \vec{a} - \vec{a}')\vec{x}      
      && \text{(by $(\lambda \vec{a} - \vec{a}') \geq \vec{0}$ and
        $\vec{x} \geq \vec{l} \geq \vec{0}$)} \\
      &= \lambda \vec{a} \vec{x} - \vec{a}' \vec{x} \\
      &\leq -\vec{a}' \vec{x}
      && \text{(by $\lambda \geq 0$ and $\vec{a} \vec{x} \leq 0$)}.
    \end{align*}

  \item $\Rightarrow$) Since $X_{<} \subseteq Y_<$, this system has no
    solution: $\exists \vec{x} : \vec{x} \geq \vec{l} \land \vec{a}
    \vec{x} < 0 \land \vec{a}' \vec{x} \geq b'$. The latter can be
    rephrased as $\exists \vec{x}, y : y \geq 1 \land \vec{x} \geq y
    \vec{l} \land \vec{a} \vec{x} \leq -1 \land \vec{a}' \vec{x} \geq y
    b'$. In matrix notation, this corresponds to $\exists \vec{x} :
    \mat{A} \vec{x} \leq \vec{c}$ where
    \[
    \mat{A} \defeq
    \begin{bmatrix}
      -\mat{I}  & \vec{l}^\transpose\ \\
      \vec{a}   & 0\\
      -\vec{a}' & b' \\
      \vec{0}   & -1
    \end{bmatrix}
    \text{ and }
    \vec{c} \defeq
    \begin{pmatrix}
      \vec{0} \\
      -1 \\
      0  \\
      -1 \\
    \end{pmatrix}.
    \]
    By Lemma~\ref{lem:farkas}, $\mat{A}^\transpose \vec{y} = \vec{0}$
    and $\vec{c}^\transpose \vec{y} < 0$ for some $\vec{y} \geq
    \vec{0}$. In other words,
    \[
    \exists \vec{z} \geq \vec{0}, \lambda, \lambda', \lambda'' \geq 0
    : -\vec{z} + \lambda \vec{a} - \lambda' \vec{a}' = \vec{0} \land
    \vec{z} \vec{l} + \lambda' b' - \lambda'' = 0 \land -\lambda -
    \lambda'' < 0.
    \]
    Since $\vec{z} \geq \vec{0}$, $\lambda \geq 0$ and $\lambda'' \geq
    0$, we have:
    \[
    \lambda \vec{a} \geq \lambda' \vec{a}' \land
    [
      (-\lambda' b' < (\lambda \vec{a} - \lambda' \vec{a}')\vec{l})
      \lor
      (-\lambda' b' = (\lambda \vec{a} - \lambda' \vec{a}')\vec{l}
      \land \lambda > 0)
    ].
    \]
    If $\lambda' > 0$, then we are done by dividing all terms by
    $\lambda'$. For the sake of contradiction, suppose that $\lambda'
    = 0$. This means that $\lambda \vec{a} \geq \vec{0}$, and either
    $0 < \lambda \vec{a} \vec{l}$ or $0 = \lambda \vec{a} \vec{\ell}
    \land \lambda > 0$. We necessarily have $\lambda > 0$ and $\vec{a}
    \geq \vec{0}$. Let $\vec{x} \in X_{<}$. We have $0 > \vec{a}
    \vec{x} \geq \vec{a} \vec{l} \geq 0$, which is a contradiction.

    \medskip
    $\Leftarrow$) Let $\vec{x} \in X_{<}$. If $-b' < (\lambda \vec{a}
    - \vec{a}') \vec{l}$ holds, then we get $\vec{x} \in Y_{<}$ as in
    Item~2. Otherwise, we have $-b' = (\lambda \vec{a} - \vec{a}')
    \vec{l}$ and $\lambda > 0$. Hence, we have:
    \begin{align*}
      -b'
      &= (\lambda \vec{a} - \vec{a}')\vec{l} \\
      &\leq (\lambda \vec{a} - \vec{a}')\vec{x}      
      && \text{(by $(\lambda \vec{a} - \vec{a}') \geq \vec{0}$ and
        $\vec{x} \geq \vec{l} \geq \vec{0}$)} \\
      &= \lambda \vec{a} \vec{x} - \vec{a}' \vec{x} \\
      &< -\vec{a}' \vec{x}
      && \text{(by $\lambda > 0$ and $\vec{a} \vec{x} < 0$)}.
    \end{align*}
    Thus, $\vec{a}' \vec{x} < b'$ and hence $\vec{x} \in Y_{<}$ as
    desired. \qedhere
  \end{enumerate}
\end{proof}

The conditions arising from Proposition~\ref{prop:incl:atomic} involve
solving linear programs with \emph{one} variable $\lambda$. It is easy
to see that this problem is in NC:

\begin{prop}\label{prop:linprog:unary}
  Given $\vec{a}, \vec{b} \in \Q^n$ and $\vec{{\sim}} \in \{{\leq},
  {<}\}^n$, testing $\exists \lambda \geq 0 : \vec{a} \lambda
  \vec{\sim} \vec{b}$ is in NC.
\end{prop}

\begin{proof}
  Let $X_i \defeq \{\lambda \geq 0 : \vec{a}_i
  \lambda~\vec{{\sim}}_i~\vec{b}_i\}$. Let ``${}_i\langle$'' and
  ``$\rangle_i$'' denote ``$[$'' and ``$]$'' if $\vec{{\sim}}_i =
         {\leq}$, and ``$($'' and ``$)$'' otherwise. Each $X_i$ is an
         interval:
  \[
  X_i =
  \begin{cases}
    \emptyset
    & \text{if } \vec{a}_i = 0 \text{ and } 0 \not\vec{\sim}_i \vec{b}_i, \\
    [0, \vec{b}_i / \vec{a}_i \rangle_i
    & \text{if } \vec{a}_i > 0, \\
    {}_i\langle \vec{b}_i / \vec{a}_i, +\infty)
    & \text{if } \vec{a}_i < 0 \text{ and } \vec{b}_i \leq 0, \\
    [0, +\infty)
      & \text{otherwise}.
  \end{cases}
  \]
  We want to test whether $X_1 \cap \cdots \cap X_n \neq \emptyset$.
  Since arithmetic belongs to NC, it suffices to: (1)~compute the left
  endpoint $\ell_i$ and the right endpoint $r_i$ of each $X_i$ in
  parallel; (2)~compute $\ell \defeq \max(\ell_1, \ldots, \ell_n)$ and
  $r \defeq \min(r_1, \ldots, r_n)$; and (3)~accept iff $\ell < r$, or
  $\ell = r$ and each $X_i$ is closed on the left and the right. 
\end{proof}

Recall that at the beginning of the section we made the assumption
that some pair $(\vec{m}, \vec{m}') \in \sol{\psi}$ is such that
$\vec{m}'$ enables a transition $t$. Checking whether this is actually
true has a cost. Fortunately, we provide a simple characterization of
enabledness which can checked in NC. Formally, we say that $\varphi$
\emph{enables} $t$ if there exists $(\vec{m}, \vec{m}') \in
\sol{\varphi}$ such that $\vec{m}'$ $\alpha$-enables $t$ for some
$\alpha > 0$. We have:

\begin{restatable}{prop}{propCharacInv}\label{prop:charac:inv}
  Let $\varphi_{\sim}(\vec{m}, \vec{m}') \defeq \vec{a} \vec{m} \sim
  \vec{b} \vec{m}'$ where $\vec{a}, \vec{b} \in \R^P$. It is the case that
  \begin{enumerate}
  \item $\varphi_{<}$ enables $u$ iff $\vec{a} \not\geq \vec{0}$ or
    $\vec{b} \not\leq \vec{0}$, and

  \item $\varphi_{\leq}$ enables $u$ iff $\vec{b} \prevec{u} \geq 0$
    or $(\vec{b} \prevec{u} < 0 \land (\vec{a}, -\vec{b}) \not\geq
    (\vec{0}, \vec{0}))$.
  \end{enumerate}
\end{restatable}

\begin{proof}
  \leavevmode
  \begin{enumerate}
  \item $\Rightarrow$) Since $\varphi_{<}$ enables $u$, we have
    $\sol{\varphi_{<}} \neq \emptyset$. Let $(\vec{m}, \vec{m}') \in
    \sol{\varphi_{<}}$. We have $\vec{a} \vec{m} < \vec{b}
    \vec{m}'$. It cannot be that $\vec{a} \geq \vec{0}$ and $\vec{b}
    \leq \vec{0}$, as otherwise $\vec{a} \vec{m} \geq 0 \geq \vec{b}
    \vec{m}'$.

    \medskip
    $\Leftarrow$) It suffices to give a pair $(\vec{m}, \vec{m}') \in
    \sol{\varphi_{<}}$ such that $\vec{m}' \geq
    \prevec{u}$. Informally, if $\vec{a}$ has a negative value
    (resp.\ $\vec{b}$ has a positive value), then we can consider the
    pair $(\vec{0}, \prevec{u})$ and ``fix'' the value on the
    left-hand-side (resp.\ right-hand side) so that $\varphi_{<}$ is
    satisfied. More formally, if $\vec{a}(p) < 0$, then $(k \vec{e}_p,
    \prevec{u}) \in \sol{\varphi_{<}}$ with $k \defeq (|\vec{b}
    \prevec{u}| + 1) / |\vec{a}(p)|$; if $\vec{b}(p) > 0$, then
    $(\vec{0}, \prevec{u} + k \vec{e}_p) \in \sol{\varphi_{<}}$ with
    $k \defeq (|\vec{b} \prevec{u}| + 1) / \vec{b}(p)$.

    \medskip
  \item $\Rightarrow$) Let $(\vec{m}, \vec{m}') \in
    \sol{\varphi_{\leq}}$ be such that $\vec{m}'$ enables $u$. We have
    $\vec{m}' = \vec{x} + \alpha \prevec{u}$ for some $\vec{x} \geq
    \vec{0}$ and $\alpha > 0$. Therefore, $\vec{a} \vec{m} \leq
    \vec{b} \vec{x} + \alpha \vec{b} \prevec{u}$. We assume that
    $\vec{b} \prevec{u} < 0$, as we are otherwise trivially done. If
    $\vec{a} \geq \vec{0}$ and $-\vec{b} \geq \vec{0}$, then we obtain
    a contradiction since $\vec{a} \vec{m} \geq 0 > \vec{b} \vec{x} +
    \alpha \vec{b} \prevec{u}$.

    \medskip
    $\Leftarrow$) It suffices to exhibit a pair $(\vec{m}, \vec{m}')
    \in \sol{\varphi_{\leq}}$ such that $\vec{m}' \geq \prevec{u}$. If
    $\vec{b} \prevec{u} \geq 0$, then we are done by taking $(\vec{0},
    \prevec{u})$. Let us consider the second case where $\vec{b}
    \prevec{u} < 0$ and $\vec{a}(p) < 0 \lor \vec{b}(p) > 0$ for some
    $p \in P$. Informally, we consider the pair $(\vec{0},
    \prevec{u})$ and ``fix'' the value on the left-hand-side or
    right-hand side depending on whether $\vec{a}(p) < 0$ or
    $\vec{b}(p) > 0$. More formally, we are done by taking either the
    pair $(k \vec{e}_p, \prevec{u})$ where $k \defeq |\vec{b}
    \prevec{u}| / |\vec{a}(p)|$, or the pair $(\vec{0}, \prevec{u} +
    \ell \vec{e}_p)$ where $\ell \defeq |\vec{b} \prevec{u}| /
    \vec{b}(p)$. \qedhere
  \end{enumerate}
\end{proof}

We can finally show that testing $\psi \leadsto_t \psi'$ can be done
in NC, for atomic propositions $\psi$ and $\psi'$. In turn, this
allows us to show that we can test in NC whether a linear formula is a
locally closed bi-separator.

\begin{prop}\label{prop:test:inv}
  Given a Petri net $\pn$, a transition $t$ and homogeneous atomic
  propositions $\psi$ and $\psi'$, testing whether
  $\psi \leadsto_t \psi'$ can be done in NC.
\end{prop}

\begin{proof}
  Recall that addition, subtraction, multiplication, division and
  comparison can be done in NC. Note that, by
  Proposition~\ref{prop:charac:inv}, we can check whether $\psi$
  enables $t$ in NC. If it does, then we must test whether $(\vec{m},
  \vec{m}') \in \sol{\psi}$ and $\vec{m}' \trans{\alpha t} \vec{m}''$
  implies $(\vec{m}, \vec{m}'') \in \sol{\psi'}$. We claim that this
  amounts to testing $X \subseteq Y$, where:
  \begin{align*}
    X
    &\defeq \{(\vec{m}, \vec{m}') \in \Rpos^P \times \Rpos^P :
    (\vec{m}, \vec{m}') \in \sol{\psi} \text{ and }
    (\vec{m}, \vec{m}') \geq (\vec{0}, \prevec{t})\}, \\
    Y &\defeq \{(\vec{m}, \vec{m}') \in \Rpos^P \times \Rpos^P :
    (\vec{m}, \vec{m}' + \effect{t}) \in \sol{\psi'}\}.
  \end{align*}
  Let us prove this claim.
  
  $\Rightarrow$) Let $(\vec{m}, \vec{m}') \in X$. We have $(\vec{m},
  \vec{m}') \in \sol{\psi}$ and $(\vec{m}, \vec{m}') \geq (\vec{0},
  \prevec{t})$. Thus $\vec{m}' \trans{t} \vec{m}' + \effect{t}$. By
  assumption, $(\vec{m}, \vec{m}' + \effect{t}) \in \sol{\psi'}$, and
  hence $(\vec{m}, \vec{m}') \in Y$.

  $\Leftarrow$) Let $(\vec{m}, \vec{m}') \in \sol{\psi}$ and $\vec{m}'
  \trans{\alpha t} \vec{m}''$. We have $\vec{m}' \geq \alpha
  \prevec{t}$ and $\vec{m}'' = \vec{m}' + \alpha \effect{t}$. Let
  $\vec{k} \defeq \vec{m} / \alpha$, $\vec{k}' \defeq \vec{m}' /
  \alpha$ and $\vec{k}'' \defeq \vec{m}'' / \alpha$. As $\alpha > 0$
  and $\psi$ is homogeneous, we have $(\vec{k}, \vec{k}') \in
  \sol{\psi}$, $(\vec{k}, \vec{k}') \geq (\vec{0}, \prevec{t})$ and
  $\vec{k}'' = \vec{k}' + \effect{t}$. Thus, $(\vec{k}, \vec{k}') \in
  X \subseteq Y$. By definition of $Y$, this means that $(\vec{k},
  \vec{k}'') \in \sol{\psi'}$. By homogeneity, we conclude that
  $(\vec{m}, \vec{m}'') \in \sol{\psi'}$.

  Now that we have shown the claim, let us explain how to check
  whether $X \subseteq Y$ in NC. Note that $X \neq \emptyset$ since
  $\psi$ enables $t$. Thus, by Proposition~\ref{prop:incl:atomic},
  testing $X \subseteq Y$ amounts to solving a linear program in one
  variable. For example, if $\psi = (\vec{a} \cdot (\vec{m}, \vec{m}')
  \leq 0)$ and $\psi' = (\vec{a}' \cdot (\vec{m}, \vec{m}') < 0)$,
  then we must check whether this system has a solution:
  \[
  \exists \lambda \geq 0 : \lambda \vec{a} \geq \vec{a}' \land \vec{a}'
  \cdot (\vec{0}, \effect{t}) < (\lambda \vec{a} - \vec{a}') \cdot
  (\vec{0}, \prevec{t}).
  \]
  Thus, by Proposition~\ref{prop:linprog:unary}, testing $X \subseteq
  Y$ can be done in NC.
\end{proof}

\begin{thm}\label{thm:verif:nc}
  Given $\pn = (P, T, F)$, $\msrc, \mtgt \in \Qpos^P$ and a formula
  $\varphi$, testing whether $\varphi$ is a locally closed
  bi-separator for $(\msrc, \mtgt)$ can be done in NC.
\end{thm}

\begin{proof}
  Recall that $\varphi = \varphi_1 \lor \cdots \lor \varphi_n$ must be
  in DNF with homogeneous atomic propositions. As arithmetic belongs
  in NC and $\varphi$ is in DNF, we can test whether $(\msrc, \msrc)
  \in \sol{\varphi}$, $(\mtgt, \mtgt) \in \sol{\varphi}$ and $(\msrc,
  \mtgt) \notin \sol{\varphi}$ in NC by evaluating $\varphi$ in
  parallel. We can further test whether $\varphi$ is locally closed by
  checking the following (which is simply the definition of ``locally
  closed''):
  \[
  \left[\bigwedge_{\substack{t \in T \\ i \in [1..n]}} \bigvee_{j \in
    [1..n]} \bigwedge_{\psi \in \varphi_i} \bigvee_{\psi' \in
      \varphi_j} \psi \leadsto_t \psi'\right]
  \land
  \left[\bigwedge_{\substack{t \in \rev{T} \\ i \in [1..n]}}
    \bigvee_{j \in [1..n]} \bigwedge_{\psi \in \varphi_i}
    \bigvee_{\psi' \in \varphi_j} \rev{\psi} \leadsto_t
    \rev{\psi'}\right].
  \]
  By Proposition~\ref{prop:test:inv}, each test $\psi \leadsto_t
  \psi'$ can be carried in NC. Therefore, we can perform all of them
  in parallel. Note that we do not have to explicitly compute the
  transpose of transitions and formulas; we can simply swap arguments.
\end{proof}

\begin{rem}
  Testing whether $\varphi$ is locally closed is even simpler if the
  tester is also given annotations indicating for every clause
  $\varphi_i$ and transition $t$ which clause $\varphi_j$ is supposed
  to satisfy $\varphi_i \leadsto_t \varphi_j$. This mapping is a
  byproduct of the procedure to compute a locally closed bi-separator,
  and so comes at no cost.
\end{rem}

\section{Bi-separators for set-to-set unreachability}\label{sec:set2set}
In some applications, one does not have to prove unreachability of one
marking, but rather of a \emph{set} of markings, usually defined by
means of some simple linear constraints. Thus, we now consider the
more general setting of ``set-to-set reachability'', i.e.\ queries of
the form $\exists \msrc \in A, \mtgt \in B : \msrc \Utrans{*} \mtgt$,
which we denote by $A \Utrans{*} B$. We focus on the case where sets
$A$ and $B$ are described by conjunctions of atomic propositions; in
other words, $A$ and $B$ are convex polytopes defined as intersections
of half-spaces. In particular, this includes ``coverability'' queries
which are important in practice, i.e.\ where $A$ is a singleton and
$B$ is of the form $\{\vec{m} : \vec{m} \geq \vec{b}\}$.

It follows from~\cite{BH17} that testing whether $A \not\Utrans{*} B$
can be done in polynomial time. In this work, we have shown that,
whenever $A$ and $B$ are singleton sets, we can validate a certificate
for $A
\not\Utrans{*} B$ in NC (and so with lower complexity). In the
general case where either $A$ or $B$ is not a singleton, we could
proceed similarly by constructing a (locally closed) bi-separator, as
defined previously, but with the natural generalization that $(\msrc,
\msrc) \in \sol{\varphi}$, $(\mtgt, \mtgt) \in \sol{\varphi}$ and
$(\msrc, \mtgt) \notin \sol{\varphi}$ \emph{for every} $\msrc \in A$
and $\mtgt \in B$.

Unfortunately, as we will show in Section~\ref{ssec:set2set:hardness},
checking that a given linear formula is a (locally closed)
bi-separator is coNP-hard. Nonetheless, in
Section~\ref{ssec:set2set:certificate}, we will show that one can
construct a locally closed bi-separator on an \emph{altered} Petri net
with a single source and target marking.

\subsection{Hardness of checking set-to-set bi-separators}
\label{ssec:set2set:hardness}


\begin{prop}
  The problem of determining whether a given linear formula $\varphi$
  is a locally closed bi-separator for convex polytopes $(A, B)$ is
  coNP-hard, even if (1)~$A$ is a singleton and $B$ is of the form
  $\{\vec{m} : \vec{m} \geq \vec{b}\}$, or (2)~vice versa.
\end{prop}

\begin{proof}[Proof of~(1)]
  We give a reduction from the problem of determining whether a DNF
  boolean formula $\psi = \bigvee_{j \in J} \psi_j(x_1, \ldots, x_m)$
  is a tautology. As in the proof of Proposition~\ref{prop:sep:coNP},
  we define $\psi_j'$ as $\psi_j$ but where literals are modified as
  follows: $x_i$ becomes $\vec{m}(x_i) > 0$, and $\neg x_i$ becomes
  $\vec{m}(x_i) \leq 0$.

  Let $\pn = (P, T, F)$ be the Petri net such that $P \defeq \{x_1,
  \ldots, x_m, y\}$, $T \defeq \{t\}$, and $t$ consumes and produces a
  token in $y$. Let
  \[
  \varphi(\vec{m}, \vec{m}') \defeq (\vec{m}'(y) \leq 0) \lor
  \bigvee_{j \in J} [\psi_j'(\vec{m}) \land \vec{m}(y) > 0].
  \]
  Let $A \defeq \{\vec{0}\}$ and $B \defeq \{\mtgt \in \Rpos^P :
  \mtgt(y) \geq 1\}$.

  Note that $\varphi$ is homogeneous. Moreover, transition $t$ leaves
  markings unchanged, i.e.\ $\vec{m}' \trans{} \vec{m}''$ implies
  $\vec{m}'' = \vec{m}'$, and likewise backward. Therefore, we
  trivially have $\theta \leadsto_t \theta$ and $\rev{\theta}
  \leadsto_t \rev{\theta}$ for every atomic proposition $\theta$ of
  $\varphi$. Further observe that $(\vec{0}, \vec{0}) \in
  \sol{\varphi}$ and $(\vec{0}, \mtgt) \notin \sol{\varphi}$ for every
  $\mtgt \in B$.

  Consequently, $\varphi$ satisfies all of the properties of a locally
  closed bi-separator for $(A, B)$, except possibly the requirement
  that $(\mtgt, \mtgt) \in \sol{\varphi}$ for every $\mtgt \in B$. The
  latter holds iff $\psi$ is a tautology (since $\varphi(\mtgt, \mtgt)
  \equiv \bigvee_{j \in J} \varphi_j'(\mtgt)$). Thus, we are done.
\end{proof}

\begin{proof}[Proof of~(2).]
  This follows immediately from the fact that $\varphi$ is a locally
  closed bi-separator for $(A, B)$ in $\pn$ iff $\rev{\varphi}$ is a
  locally closed bi-separator for $(B, A)$ in $\rev{\pn}$.
\end{proof}

\subsection{Certifying set-to-set unreachability}
\label{ssec:set2set:certificate}

As shown in~\cite[Lem.~3.7]{BH17}, given an atomic proposition $\psi =
(\vec{a} \vec{x} \sim b)$, one can construct (in logarithmic space) a
Petri net $\pn_\psi$ and some $\vec{y} \in \{0, 1\}^5$ such that
$\psi(\vec{x})$ holds iff $(\vec{x}, \vec{y}) \Utrans{*} (\vec{0},
\vec{0})$ in $\pn_\psi$. The idea---depicted in Figure~\ref{fig:convex:cpn}, which is adapted from~\cite[Fig.~1]{BH17})---is simply to cancel out
positive and negative coefficients of $\psi$. It is straightforward to
adapt this construction to a conjunction $\bigwedge_{1 \leq i \leq k}
\psi_k(\vec{x})$ of atomic propositions. Indeed, it suffices to make
$k$ copies of the gadget, but where places $\{p_1, \ldots, p_n\}$ and
transitions $\{t_1, \ldots, t_n\}$ are shared. In this more general
setting, $t_i$ consumes from $p_i$ and simultaneously spawns the
respective coefficient to each copy.

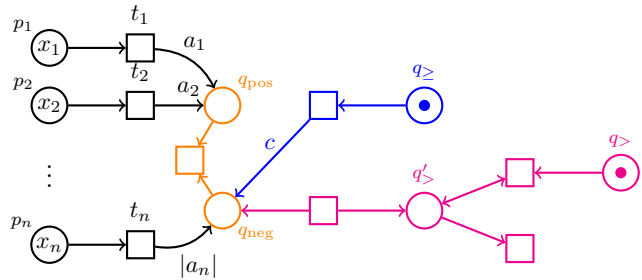
\begin{figure*}[!h]
  \centering
  \begin{tikzpicture}[node distance=0.85cm, auto, thick, scale=0.9, transform shape]
    \tikzset{every place/.style={inner sep=1pt, minimum size=15pt}}
    
    \node[place] (p1) {\small$x_1$};
    \node[place] (p2) [below of=p1] {\small$x_2$};
    \node[draw=none] (pi) [below=0.25 cm of p2] {$\vdots$};
    \node[place] (pn) [below=1.5cm of p2] {\small$x_n$};

    \node[place, colEmph3] (qpos) [right=2cm of p2] {};
    \node[place, colEmph3] (qneg) [below=1cm of qpos] {};

    \node[transition, label=above:\small$t_1$] (p1-qpos) [right=of p1] {};
    \node[transition, label=above:\small$t_2$] (p2-qpos) [right=of p2] {};
    \node[transition, label=above:\small$t_n$] (pn-qneg) [right=of pn] {};

    \node[transition, rotate=90, colEmph3] (normalize) [below left=0.4cm and 0.5cm of qpos] {};

    \node[transition, colEmph] (qgeq-qneg) [right=1cm of qpos] {};
    \node[transition, colEmph2] (qg-qneg) [right=1cm of qneg] {};

    \tikzset{every token/.style={color=colEmph}}
    \node[place, tokens=1, colEmph] (qgeq) [right=1cm of qgeq-qneg] {};
    \node[place, colEmph2] (qg) [right=1cm of qg-qneg] {};
    \node[transition, colEmph2] (qgg-qg) [above right=0.15cm and 1cm of qg] {};
    \node[transition, colEmph2] (qg-lossy) [below right=0.15cm and 1cm of qg] {};

    \tikzset{every token/.style={color=colEmph2}}
    \node[place, tokens=1, colEmph2] (qgg) [right=1cm of qgg-qg] {};

    \node[draw=none] () [above left=-3pt and -3pt of p1] {\scriptsize$p_1$};
    \node[draw=none] () [above left=-3pt and -3pt of p2] {\scriptsize$p_2$};
    \node[draw=none] () [above left=-3pt and -3pt of pn] {\scriptsize$p_n$};
    \node[draw=none, colEmph3] () [above right=-3pt and -3pt of qpos] {\scriptsize$q_\text{pos}$};
    \node[draw=none, colEmph3] () [below right=-3pt and -3pt of qneg] {\scriptsize$q_\text{neg}$};
    \node[draw=none, colEmph] () [above=-1pt of qgeq] {\scriptsize$q_{\geq}$};
    \node[draw=none, colEmph2] () [above=-1pt of qg] {\scriptsize$q_{>}'$};
    \node[draw=none, colEmph2] () [above=-1pt of qgg] {\scriptsize$q_{>}$};

    \path[->]
    (p1) edge node {} (p1-qpos)
    (p2) edge node {} (p2-qpos)
    (pn) edge node {} (pn-qneg)
    (p1-qpos) edge[bend left] node[xshift=-7pt] {\small$a_1$}   (qpos)
    (p2-qpos) edge node[xshift=4pt] {\small$a_2$}   (qpos)
    (pn-qneg) edge[bend right] node[swap, xshift=-7pt] {\small$|a_n|$} (qneg)

    (qpos) edge[colEmph3] node {} (normalize)
    (qneg) edge[colEmph3] node {} (normalize)

    (qgeq) edge[colEmph] node {} (qgeq-qneg)
    (qg-qneg) edge[colEmph2] node {} (qg)

    (qgeq-qneg) edge[colEmph] node[swap, xshift=5pt] {\small$c$} (qneg)
    (qg-qneg) edge[colEmph2] node {} (qneg)

    (qgg) edge[colEmph2] node {} (qgg-qg)
    (qg) edge[colEmph2] node {} (qg-lossy)
    ;

    \path[<->]
    (qgg-qg) edge[colEmph2] node {} (qg)
    ;    
  \end{tikzpicture}
  \caption{Petri net for $\psi(\vec x) = (a_1 \cdot x_1 + \cdots +
    a_n\cdot x_n > c)$ where $a_1, a_2, c > 0$ and $a_n <
    0$.}\label{fig:convex:cpn}
\end{figure*}

In summary, the following holds:

\begin{prop}\label{prop:cpn:to:convex}
  Given a conjunction of $k$ atomic propositions $\varphi$, it is
  possible to construct, in logarithmic space, a Petri net
  $\pn_\varphi$ and $\vec{y} \in \{0, 1\}^{5k}$ such that
  $\varphi(\vec{x})$ holds iff $(\vec{x}, \vec{y}) \Utrans{*}
  (\vec{0}, \vec{0})$ in $\pn_\varphi$.
\end{prop}

With the previous construction in mind, we can reformulate any
set-to-set reachability query into a standard (marking-to-marking)
reachability query.

\begin{prop}\label{prop:convex:cpn}
  Given a Petri net $\pn$ and convex polytopes $A$ and $B$ described
  as conjunctions of atomic propositions, one can construct, in
  logarithmic space, a Petri net $\pn'$ and markings $\msrc$ and
  $\mtgt$ such that $A \Utrans{*} B$ in $\pn$ iff $\msrc \Utrans{*}
  \mtgt$ in $\pn'$.
\end{prop}

\begin{proof}
  Let $\pn = (P, T, (\mat{F}_-, \mat{F}_+))$ where $P = \{p_1, \ldots,
  p_n\}$. Let us describe $\pn' = (P', T', \allowbreak (\mat{F}_-',
  \allowbreak \mat{F}_+'))$ with the help of
  Figure~\ref{fig:set2set}. The Petri net $\pn'$ extends $\pn$ as
  follows:
  \begin{itemize}
  \item we add transitions $\{t_1, \ldots, t_n\}$ whose purpose is to
    nondeterministically guess an initial marking of $\pn$ in $P$, and
    make a copy in $Q \defeq \{q_1, \ldots, q_n\}$;

  \item we add a net $\pn_A = (P_A, T_A, F_A)$, obtained from
    Proposition~\ref{prop:cpn:to:convex}, to test whether the marking
    in $Q$ belongs to $A$; and we add a net $\pn_B = (P_B, T_B, F_B)$,
    obtained from Proposition~\ref{prop:cpn:to:convex}, to test
    whether the marking in $P$ belongs to $B$.
  \end{itemize}

  \begin{figure*}
    \centering
    \begin{tikzpicture}[node distance=1cm, auto, thick, scale=0.9, transform shape]
      \tikzset{every place/.style={inner sep=1pt, minimum size=15pt}}

      \node[place, label=above:$q_1$]              (p1) {};
      \node[place, right of=p1, label=above:$q_2$] (p2) {};
      \node[right of=p2] (d) {$\cdots$};
      \node[place, right of=d, label=above:$q_n$]  (pn) {};

      \node[place, label=below:$p_1$, below=1.25cm of p1]  (p1p) {};
      \node[place, right of=p1p, label=below:$p_2$] (p2p) {};
      \node[right of=p2p] (dp) {$\cdots$};
      \node[place, right of=dp, label=below:$p_n$]  (pnp) {};

      \node[draw, minimum width=70pt, minimum height=25pt,
            right=2cm of pn, yshift=5pt] (A) {$\in A$?};

      \path[->]
      (p1) edge[out=-50, in=180] node {} (A)
      (p2) edge[out=-35, in=180] node {} (A)
      (pn) edge[out=-10, in=180] node {} (A)
      ;
      
      \node[draw, minimum width=70pt, minimum height=25pt,
            right=2cm of pnp, yshift=-5pt] (B) {$\in B$?};

      \path[->]
      (p1p) edge[out=50, in=180] node {} (B)
      (p2p) edge[out=35, in=180] node {} (B)
      (pnp) edge[out=10, in=180] node {} (B)
      ;

      \node[transition, label=left:$t_1$, left=1.5cm of p1,
            yshift=15pt] (t1) {};
      \node[transition, below of=t1, label=left:$t_2$] (t2) {};
      \node[below of=t2]                               (dt) {$\vdots$};
      \node[transition, below of=dt, label=left:$t_n$] (tn) {};

      \path[->, gray, line width=0.5pt]
      (t1) edge[out=0,  in=180]  node {} (p1)
      (t1) edge[out=0,  in=180]  node {} (p1p)
      (t2) edge[out=0,  in=-135] node {} (p2)
      (t2) edge[out=0,  in=135]  node {} (p2p)
      (tn) edge[out=30, in=180]  node {} (pn)
      (tn) edge[out=0,  in=-135] node {} (pnp)
      ;

      \node[below=0.75cm of p1p, xshift=1.45cm] (N) {$\pn$};
      
      \begin{pgfonlayer}{back}
        \node[fit=(p1p)(p2p)(dp)(pnp)(N), fill=colEmph3,
              fill opacity=0.25, inner sep=7pt] {};
      \end{pgfonlayer}
    \end{tikzpicture}
    \caption{Reduction from set-to-set reachability to
      (marking-to-marking) reachability.}\label{fig:set2set}
  \end{figure*}
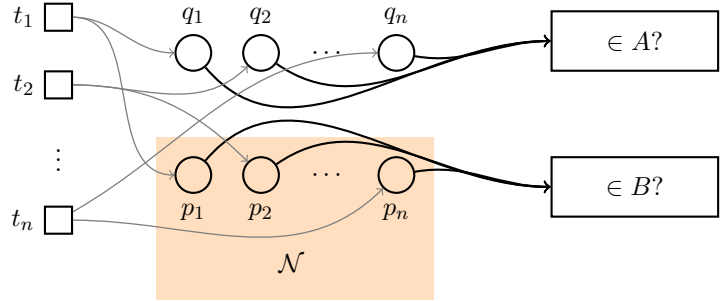
  
  The Petri net $\pn'$ is \emph{intended} to work sequentially as
  follows:
  \begin{enumerate}
  \item guess the initial marking $\msrc$ of $\pn$;

  \item test whether $\msrc \in A$;

  \item execute $\pn$ on $\msrc$ and reach a marking $\mtgt$; and

  \item test whether $\mtgt \in B$.
  \end{enumerate}

  If $\pn'$ follows this order, then it is straightforward to see that
  $A \Utrans{*} B$ in $\pn$ iff $(\vec{0}, \vec{0}, \vec{y}, \vec{y}')
  \Utrans{*} (\vec{0}, \vec{0}, \vec{0}, \vec{0})$ in $\pn'$, where
  $\vec{y}$ and $\vec{y}'$ are obtained from
  Proposition~\ref{prop:cpn:to:convex}. However, $\pn'$ may interleave
  the different phases.\footnote{It is tempting to implement a lock,
    but this only works under discrete semantics.}  Nonetheless, this
  is not problematic, as any run of $\pn'$ can be reordered in such a
  way that all four phases are consecutive.

  More formally, let $(\vec{0}, \vec{0}, \vec{y}, \vec{y}')
  \trans{\sigma} (\vec{0}, \vec{0}, \vec{0}, \vec{0})$ in $\pn'$. Note
  that $\prevec{t_i} = \vec{0}$ for every $i \in [1..n]$. Thus, we can
  reorder $\sigma$ so that
  \[
  (\vec{0}, \vec{0}, \vec{y}, \vec{y}')
  \trans{\alpha_1 t_1 \cdots \alpha_n t_n}
  (\msrc, \msrc, \vec{y}, \vec{y}')
  \trans{\sigma'}
  (\vec{0}, \vec{0}, \vec{0}, \vec{0}) \text{ in } \pn',
  \]
  for some $\alpha_1, \ldots, \alpha_n \in \Rpos$, marking $\msrc$,
  and firing sequence $\sigma'$ whose support does not contain any of
  $t_1, \ldots, t_n$.

  We have $\prevec{t}(r) = 0$ for every $t \in T_A$ and $r \notin Q
  \cup P_A$. Furthermore, we have $\postvec{t}(r) = 0$ for every $t
  \notin \{t_1, \ldots, t_n\}$ and $r \in Q \cup P_A$. Thus, we can
  reorder $\sigma'$ into $\tau \tau'$ so that
  \[
  (\msrc, \msrc, \vec{y}, \vec{y}')
  \trans{\tau}
  (\vec{0}, \msrc, \vec{0}, \vec{y}')
  \trans{\tau'}
  (\vec{0}, \vec{0}, \vec{0}, \vec{0}) \text{ in } \pn',
  \]
  the support of $\tau$ only contains transitions from $T_A$, no
  transition from $T_A$ occurs in the support of $\tau'$.

  We have $\postvec{t}(r) = 0$ for every $t \in T_B$ and $r \notin
  P_B$. Therefore, we can reorder $\tau'$ into $\tau'' \tau'''$ so
  that
  \[
  (\vec{0}, \msrc, \vec{0}, \vec{y}')
  \trans{\tau''}
  (\vec{0}, \mtgt, \vec{0}, \vec{y}')
  \trans{\tau'''}
  (\vec{0}, \vec{0}, \vec{0}, \vec{0}) \text{ in } \pn',
  \]
  the support of $\tau''$ only contains transitions from $T$, no
  transition from $T$ occurs in the support of $\tau'''$, and $\mtgt
  \in \Rpos^P$.

  Altogether, we obtain
  \[
  (\msrc, \vec{y}) \trans{\tau} (\vec{0}, \vec{0}) \text{ in } \pn_A,
  \msrc \trans{\tau'} \mtgt \text{ in } \pn, \text{ and }
  (\mtgt, \vec{y}') \trans{\tau''} (\vec{0}, \vec{0}) \text{ in } \pn_B.
  \]
  From this, we conclude that $\msrc \in A$, $\mtgt \in B$ and $\msrc
  \Utrans{*} \mtgt$ in $\pn$.
\end{proof}

As a consequence of Proposition~\ref{prop:convex:cpn}, combined with
Theorems~\ref{thm:unreach:certificate} and~\ref{thm:verif:nc}, we
obtain the following corollary:

\begin{cor}
  A negative answer to a convex polytope query $A \Utrans{*} B$ is
  witnessed by a locally closed bi-separator, for an altered Petri
  net, computable in polynomial time and checkable in NC.
\end{cor}

\section{Conclusion}
We have shown that continuous Petri nets admit locally closed bi-separators that can be efficiently computed. These separators are succinct and very efficiently checkable certificates of unreachability. In particular, checking that  a linear formula is a locally closed bi-separator is in NC, and only requires to solve linear inequations in one variable over the nonnegative reals. While this does not directly hold for the more general of set-to-set reachability, we have shown that it can be extended to an altered Petri net (in the case of convex polytopes).

Verification tools that have not been formally verified, or rely (as is usually the case) on external packages for linear arithmetic, can apply our results to provide certificates for their output. Further, our separators can be used as explanations of why a certain marking is unreachable. Obtaining minimal explanations is an interesting research avenue.

From a logical point of view, separators are very closely related to interpolants for linear arithmetic, which are widely used in formal verification to refine abstractions in the CEGAR approach~\cite{BeyerZM08,RybalchenkoS10,SchollPDA14,AlthausBKS15}. We intend to explore whether they can constitute the basis of a CEGAR approach for the verification of continuous Petri nets. 


\section*{Acknowledgment}

\noindent We thank the anonymous referees of LMCS and FoSSaCS 2022 for
their comments, and in particular for suggesting a more intuitive
definition of bi-separators.

\bibliographystyle{alphaurl}
\bibliography{references}

\appendix
\section{Missing proofs}
\begin{proof}[Proof of Lemma~\ref{prop:farkas:supp}.]
  Let $\mathcal{S}$ and $\mathcal{S}'$ denote the two systems of the
  proposition. We must show that $\mathcal{S}$ has no solution iff
  $\mathcal{S}'$ has a solution.  First, we prove the following claim:
  
  \medskip \noindent \textit{Claim}:  The system $\exists \vec{x} \geq \vec{0} : \mat{A} \vec{x} = \vec{b}
  \land S \subseteq \supp{\vec{x}} \subseteq S'$ has a solution iff
  this system has a solution: $\exists \vec{x} \geq \vec{0}, y \geq 1
  : \mat{A} \vec{x} = y \vec{b} \land \vec{x} =_{\,\overline{V}}
  \vec{0} \land \vec{x} \geq_U \vec{1}$, where $\vec{1} \defeq (1,
  \ldots, 1)$.
  
  \medskip \noindent \textit{Proof of the claim}: $\Rightarrow$) Since
  $\vec{x}(t) > 0$ for all $t \in S$, we can pick $y \geq 1$
  sufficiently large so that $y \vec{x}(t) \geq 1$ for every $t \in
  S$. Let $\vec{x}' \defeq y \vec{x}$. We have $\mat{A} \vec{x}' =
  y \mat{A} \vec{x} = y \vec{b}$ and $\vec{x}' =
  y \vec{x} \geq \vec{x} \geq \vec{0}$. Moreover, $\vec{x}'(t) =
  y \vec{x}(t) = 0$ for every $t \notin S'$, and $\vec{x}'(t) =
  y \vec{x}(t) \geq 1$ for every $t \in S$.

  $\Leftarrow$) Let $\vec{x}' \defeq \vec{x} / y$. We have $\mat{A}
  \vec{x}' = \vec{b}$. Moreover, for every $t \notin S'$ it is the case that $\vec{x}'(t) = (1/y)\vec{x}(t)t = 0$ and for every $ts \in
  S$ it is the case that  $\vec{x}'(t) = (1/y) \vec{x}(t) \geq 1/y > 0$. Hence, $S \subseteq
  \supp{\vec{x}'} \subseteq S'$.
  
  \medskip Now we proceed to prove the proposition. Let $\mat{J} \in \R^{\overline{S'}
    \times T}$ be the matrix that contains $0$ everywhere except for
  $\mat{J}_{t, t} = 1$ for all $t \notin S'$. Let $\vec{c} \in \R^T$
  be such that $\vec{c}(t) = 1$ for every $t \in S$ and $\vec{c}(t) = 0$ for every $t \notin
  S$. By the claim above, the system
  $\mathcal{S}$ has a solution iff the system $\exists \vec{x}' :
  \mat{A}' \vec{x}' \leq \vec{b}'$ has a solution, where
  \[
  \mat{A}' \defeq
  \begin{pmatrix}
    \mat{A}  & -\vec{b} \\
    -\mat{A} & \vec{b}  \\
    \mat{J} & \vec{0}  \\
    -\mat{I} & \vec{0}  \\
    \vec{0}^\transpose & -1
  \end{pmatrix}
  \text{ and }
  \vec{b}' \defeq
  \begin{pmatrix}
    \vec{0} \\
    \vec{0} \\
    \vec{0} \\
    -\vec{c} \\
    -1
  \end{pmatrix}.
  \]
  By Lemma~\ref{lem:farkas}, the latter system has no solution iff the
  following has one: $\exists \vec{z} \geq \vec{0} :
  (\mat{A}')^\transpose \vec{z} = \vec{0} \land (\vec{b}')^\transpose
  \vec{z} < 0$. We can rewrite the latter as follows:
  \begin{alignat*}{2}
  & \exists \vec{z} \geq \vec{0} :
  (\mat{A}')^\transpose \vec{z} = \vec{0}
  \land (\vec{b}')^\transpose \vec{z} < 0 \\
  \equiv\ &
  \exists (\vec{u}, \vec{u}', \vec{v}, \vec{v}', \alpha) \geq \vec{0} :
  \mat{A}^\transpose(\vec{u} - \vec{u}') =
  \vec{v}' - \mat{J}^\transpose \vec{v} \land
  \vec{b}^\transpose(\vec{u}' - \vec{u}) = \alpha \land
  -\vec{c}^\transpose \vec{v}'  - \alpha < 0 \\
  \equiv\ &
  \exists \vec{y}\ \exists (\vec{v}, \vec{v}', \alpha) \geq \vec{0} :
  \mat{A}^\transpose \vec{y} =
  \vec{v}' - \mat{J}^\transpose \vec{v} \land
  \vec{b}^\transpose\vec{y} = -\alpha \land
  -\alpha < \vec{c}^\transpose \vec{v}' \\
  \equiv\ &
  \exists \vec{y}\ \exists (\vec{v}, \vec{v}') \geq \vec{0} :
  \mat{A}^\transpose \vec{y} =
  \vec{v}' - \mat{J}^\transpose \vec{v} \land
  \vec{b}^\transpose\vec{y} \leq 0 \land
  \vec{b}^\transpose\vec{y} < \vec{c}^\transpose \vec{v}' \\
  \equiv\ &
  \exists \vec{y} :
  \mat{A}^\transpose \vec{y} \geq_{S'} \vec{0} \land
  \vec{b}^\transpose\vec{y} \leq 0 \land
  \vec{b}^\transpose\vec{y} < \sum_{s \in S} (\mat{A}^\transpose \vec{y})_s.
  \end{alignat*}
  Note that the last equivalence holds since $(\mat{J}^\transpose
  \vec{v}) =_{S'} \vec{0}$ and since $\vec{c}^\transpose$ sums entries
  over $S$. 
\end{proof}

\end{document}